\documentclass[conference, a4paper]{article}
\usepackage[numbers]{natbib}
\usepackage{eqnarray}
\usepackage{xcolor}

\usepackage{amsmath,amssymb,amsthm}
\usepackage{etoolbox}
\usepackage{tikz}
\usetikzlibrary{positioning,
    calc,
    arrows,
    decorations.markings,
    decorations.pathreplacing,
    backgrounds,
    patterns,
    intersections}

\pgfdeclarelayer{background}
\pgfdeclarelayer{foreground}
\pgfsetlayers{background,main,foreground}

\newcommand\ignore[1]{\ignorespacesafterend}
\newtheorem{theorem}{Theorem}
\numberwithin{theorem}{section}

\newtheorem{lemma}[theorem]{Lemma}
\theoremstyle{definition}

\allowdisplaybreaks[1]

\makeatletter
\def\calign@preamble{%
   &\hfil\strut@
    \setboxz@h{\@lign$\m@th\displaystyle{##}$}%
    \ifmeasuring@\savefieldlength@\fi
    \set@field
    \hfil
    \tabskip\alignsep@
}
\let\cmeasure@\measure@
\patchcmd\cmeasure@{\divide\@tempcntb\tw@}{}{}{}
\patchcmd\cmeasure@{\divide\@tempcntb\tw@}{}{}{}
\patchcmd\cmeasure@{\ifodd\maxfields@
  \global\advance\maxfields@\@ne
  \fi}{}{}{}    
\newenvironment{calign}
{%
  \let\align@preamble\calign@preamble
  \let\measure@\cmeasure@
  \align
}
{%
  \endalign
}  
\makeatother

\allowdisplaybreaks[1]

\newcommand\id{\mathrm{id}}
\newcommand\Hom{\mathrm{Hom}}
\newcommand\cat[1]{\ensuremath{\mathbf{#1}}}

\def\fillA{blue!20}

\tikzset{keyvertexcolour/.initial=black}
\tikzset{vertex colour/.style={keyvertexcolour={#1}}}
\newlength\vertexradius
\newlength\innerradius
\def\halfanglesep{24}
\def\tripleanglesep{24}
\def\sideangle{30}
\def\tripleanglesep{40}
\makeatletter
\pgfdeclareshape{Vertex}
{
    \savedanchor\centerpoint
    {
        \pgf@x=0pt
        \pgf@y=0pt
    }
    \anchor{n1}
    {
        \pgfextractx{\pgf@x}{\pgfpointpolar{90+\halfanglesep}{\innerradius}}
        \pgfextracty{\pgf@y}{\pgfpointpolar{90+\halfanglesep}{\innerradius}}
    }
    \anchor{n2}
    {
        \pgfextractx{\pgf@x}{\pgfpointpolar{90-\halfanglesep}{\innerradius}}
        \pgfextracty{\pgf@y}{\pgfpointpolar{90-\halfanglesep}{\innerradius}}
    }
    \anchor{n}
    {
        \pgfextractx{\pgf@x}{\pgfpointpolar{90}{\innerradius}}
        \pgfextracty{\pgf@y}{\pgfpointpolar{90}{\innerradius}}
    }
    \anchor{ne1}
    {
        \pgfextractx{\pgf@x}{\pgfpointpolar{\sideangle+\halfanglesep}{\innerradius}}
        \pgfextracty{\pgf@y}{\pgfpointpolar{\sideangle+\halfanglesep}{\innerradius}}
    }
    \anchor{ne2}
    {
        \pgfextractx{\pgf@x}{\pgfpointpolar{\sideangle-\halfanglesep}{\innerradius}}
        \pgfextracty{\pgf@y}{\pgfpointpolar{\sideangle-\halfanglesep}{\innerradius}}
    }
    \anchor{ne}
    {
        \pgfextractx{\pgf@x}{\pgfpointpolar{\sideangle}{\innerradius}}
        \pgfextracty{\pgf@y}{\pgfpointpolar{\sideangle}{\innerradius}}
    }
    \anchor{se1}
    {
        \pgfextractx{\pgf@x}{\pgfpointpolar{-\sideangle+\halfanglesep}{\innerradius}}
        \pgfextracty{\pgf@y}{\pgfpointpolar{-\sideangle+\halfanglesep}{\innerradius}}
    }
    \anchor{se2}
    {
        \pgfextractx{\pgf@x}{\pgfpointpolar{-\sideangle-\halfanglesep}{\innerradius}}
        \pgfextracty{\pgf@y}{\pgfpointpolar{-\sideangle-\halfanglesep}{\innerradius}}
    }
    \anchor{s1}
    {
        \pgfextractx{\pgf@x}{\pgfpointpolar{-90-\halfanglesep}{\innerradius}}
        \pgfextracty{\pgf@y}{\pgfpointpolar{-90-\halfanglesep}{\innerradius}}
    }
    \anchor{s2}
    {
        \pgfextractx{\pgf@x}{\pgfpointpolar{-90+\halfanglesep}{\innerradius}}
        \pgfextracty{\pgf@y}{\pgfpointpolar{-90+\halfanglesep}{\innerradius}}
    }
    \anchor{sA}
    {
        \pgfextractx{\pgf@x}{\pgfpointpolar{-90-\tripleanglesep}{\innerradius}}
        \pgfextracty{\pgf@y}{\pgfpointpolar{-90-\tripleanglesep}{\innerradius}}
    }
    \anchor{sB}
    {
        \pgfextractx{\pgf@x}{\pgfpointpolar{-90}{\innerradius}}
        \pgfextracty{\pgf@y}{\pgfpointpolar{-90}{\innerradius}}
    }
    \anchor{sC}
    {
        \pgfextractx{\pgf@x}{\pgfpointpolar{-90+\tripleanglesep}{\innerradius}}
        \pgfextracty{\pgf@y}{\pgfpointpolar{-90+\tripleanglesep}{\innerradius}}
    }
    \anchor{s}
    {
        \pgfextractx{\pgf@x}{\pgfpointpolar{-90}{\innerradius}}
        \pgfextracty{\pgf@y}{\pgfpointpolar{-90}{\innerradius}}
    }
    \anchor{sw1}
    {
        \pgfextractx{\pgf@x}{\pgfpointpolar{180+\sideangle+\halfanglesep}{\innerradius}}
        \pgfextracty{\pgf@y}{\pgfpointpolar{180+\sideangle+\halfanglesep}{\innerradius}}
    }
    \anchor{sw2}
    {
        \pgfextractx{\pgf@x}{\pgfpointpolar{180+\sideangle-\halfanglesep}{\innerradius}}
        \pgfextracty{\pgf@y}{\pgfpointpolar{180+\sideangle-\halfanglesep}{\innerradius}}
    }
    \anchor{sw}
    {
        \pgfextractx{\pgf@x}{\pgfpointpolar{180+\sideangle}{\innerradius}}
        \pgfextracty{\pgf@y}{\pgfpointpolar{180+\sideangle}{\innerradius}}
    }
    \anchor{nw1}
    {
        \pgfextractx{\pgf@x}{\pgfpointpolar{180-\sideangle+\halfanglesep}{\innerradius}}
        \pgfextracty{\pgf@y}{\pgfpointpolar{180-\sideangle+\halfanglesep}{\innerradius}}
    }
    \anchor{nw2}
    {
        \pgfextractx{\pgf@x}{\pgfpointpolar{180-\sideangle-\halfanglesep}{\innerradius}}
        \pgfextracty{\pgf@y}{\pgfpointpolar{180-\sideangle-\halfanglesep}{\innerradius}}
    }
    \anchor{nw}
    {
        \pgfextractx{\pgf@x}{\pgfpointpolar{180-\sideangle}{\innerradius}}
        \pgfextracty{\pgf@y}{\pgfpointpolar{180-\sideangle}{\innerradius}}
    }
    \anchor{center}{\centerpoint}
    \anchorborder{\centerpoint}
    \backgroundpath
    {
        \pgfkeysgetvalue{/tikz/keyvertexcolour}{\vcol}
        \begin{pgfonlayer}{foreground}
        \pgfsetfillcolor{\vcol}
        \pgfpathcircle{\pgfpoint{0cm}{0cm}}{\vertexradius}
        \pgfusepath{fill}
        \pgfsetstrokecolor{black}
        \pgfsetlinewidth{0.8pt}
        \pgfpathcircle{\pgfpoint{0cm}{0cm}}{\vertexradius}
        \pgfusepath{stroke}
        \end{pgfonlayer}
    }
}
\makeatother

\def\nwangle{180-\sideangle}
\def\neangle{\sideangle}

\newbool{partpath}
\pgfkeys{/tikz/double partial rendering/.style={
    postaction={
      decorate,
      decoration={
        show path construction,
        moveto code={
          \global\booltrue{partpath}
        },
        lineto code={
          \ifx\tikzinputsegmentfirst\tikzinputsegmentlast
            \ifbool{partpath}{\global\boolfalse{partpath}}{\global\booltrue{partpath}}
          \fi
          \ifbool{partpath}{
            \draw [double] (\tikzinputsegmentfirst) -- (\tikzinputsegmentlast);
          }{}
        },
        curveto code={
          \ifbool{partpath}{
            \draw [double] (\tikzinputsegmentfirst) .. controls (\tikzinputsegmentsupporta) and (\tikzinputsegmentsupportb) .. (\tikzinputsegmentlast);
          }{}
        },
        closepath code={
          \ifbool{partpath}{
            \draw (\tikzinputsegmentfirst) -- (\tikzinputsegmentlast);
          }{}
        }
      }
    }
  }
}
\pgfkeys{/tikz/partial rendering/.style={
    postaction={
      decorate,
      decoration={
        show path construction,
        moveto code={
          \global\booltrue{partpath}
        },
        lineto code={
          \ifx\tikzinputsegmentfirst\tikzinputsegmentlast
            \ifbool{partpath}{\global\boolfalse{partpath}}{\global\booltrue{partpath}}
          \fi
          \ifbool{partpath}{
            \draw (\tikzinputsegmentfirst) -- (\tikzinputsegmentlast);
          }{}
        },
        curveto code={
          \ifbool{partpath}{
            \draw (\tikzinputsegmentfirst) .. controls (\tikzinputsegmentsupporta) and (\tikzinputsegmentsupportb) .. (\tikzinputsegmentlast);
          }{}
        },
        closepath code={
          \ifbool{partpath}{
            \draw (\tikzinputsegmentfirst) -- (\tikzinputsegmentlast);
          }{}
        }
      }
    }
  }
}

\def\diagramscale{0.9}
\def\componentwidth{1.2cm}

\definecolor{green}{rgb}{0.2,1,0.2}

\usepackage[normalem]{ulem}

\begin{document}

\title{Bicategorical Semantics for
\\
Nondeterministic Computation}

\def\extragap{2pt}
\author{
\normalsize
\begin{tabular}{cc}
\begin{tabular}{c}
Mike Stay
\\
\texttt{stay@google.com}
\\[\extragap]
Google Inc.
\\
California, USA
\\[\extragap]
Department of Computer Science
\\
University of Auckland, New Zealand
\end{tabular}
&
\begin{tabular}{c}
Jamie Vicary
\\
\texttt{jamie.vicary@cs.ox.ac.uk}
\\[\extragap]
Centre for Quantum Technologies
\\
National University of Singapore
\\[\extragap]
Department of Computer Science
\\
University of Oxford, UK
\end{tabular}
\end{tabular}}

\date{January 15, 2013}

\maketitle

\begin{abstract}
We outline a bicategorical syntax for the interaction between public and private information in classical information theory. We use this to give high-level graphical definitions of encrypted communication and secret sharing protocols, including a characterization of their security properties. Remarkably, this makes it clear that the protocols have an identical abstract form to the quantum teleportation and dense coding procedures, yielding evidence of a deep connection between classical and quantum information processing. We also formulate public-key cryptography using our scheme. Specific implementations of these protocols as nondeterministic classical procedures are  recovered by applying our formalism in a symmetric monoidal  bicategory of matrices of relations.
\end{abstract}

\section{Introduction}

\subsection{Background}

\noindent Whitehead credited Hamilton and De Morgan with the invention of ``universal algebra'', the idea that we can describe many mathematical structures as sets equipped with functions that are subject to equations~\cite{g68-ua}.  Modern object-oriented programming is done in essentially the same way: to define a data structure, we equip a type with methods and insist that implementations pass a test suite.  The programming language gives us a syntax to express an interface as well as a way to write implementations, each of which picks out a different semantics for the interface.  Lambek~\cite{l80-flc} showed that such syntactic descriptions of interfaces correspond to free cartesian closed categories, and implementations are simply cartesian closed functors from those syntactic categories to the category of sets and functions. The related area of information flow~\cite{w05-cae} is the application of type theory to security.  The types correspond to security levels like ``public'' and ``private'', and a well-typed program is a proof that an attacker cannot distinguish two computations from their outputs if they only vary in their private inputs.  Such a derivation system corresponds to a cartesian closed category.

Similarly, in quantum information theory, a categorical approach developed initially by Abramsky and Coecke~\cite{ac04-csqp, ac08-cqm, hv13-cqm} has been shown to be extremely fruitful, based on the category of finite-dimensional Hilbert spaces. And in physics, Feynman diagrams also follow this pattern, except instead of using a category presented syntactically, it uses a category presented graphically.

The mathematical notion underlying all these areas is that of symmetric monoidal category, now widely recognized as an important unifying concept~\cite{bs-rosetta}. Given this observation, it is natural to consider what role can be played by symmetric monoidal \emph{bicategories} in the description of classical phenomena in computer science. Bicategories are algebraic structures with an extra layer of descriptive power compared to ordinary categories, and have already been demonstrated to be of importance in quantum field theory~\cite{bd95-hda} and quantum information~\cite{v12-hsqp}, where their key strength lies in their ability to encode important connectivity and locality information in a natural way.

\subsection{Overview}

We propose a bicategorical syntax for reasoning about cryptographic processes in classical computation. The extra structure of our higher syntax provides a geometrical mechanism for distinguishing public and private information, and also their interactions, including publication, privatization, copying and information retrieval processes.

Since bicategories have a well-studied 2\-dimensional graphical calculus, this becomes available for the description of our classical computational processes, and gives a powerful and elegant formalism with which to reason about them. A particular diagram can he interpreted as a history of computational events, with the vertical direction representing time, which flows from bottom to top. To use the terminology of physics, they are `spacetime diagrams' for our computation. For example, the following diagram represents an encrypted communication protocol making use of a one-time pad:
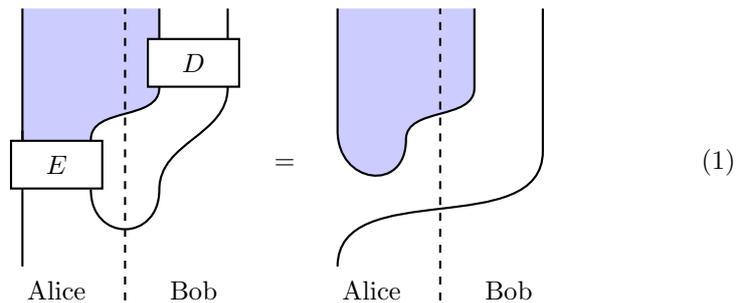
\begin{equation}
\label{eq:introencryptedcommunication}
\begin{aligned}
\begin{tikzpicture}[scale=\diagramscale, thick]
    \begin{pgfonlayer}{foreground}
    \node (U) [minimum width=\componentwidth, minimum height=18pt, fill=white, draw=black, thick] at (3,2.5) {$D$};
    \node (M) [minimum width=\componentwidth, minimum height=18pt, fill=white, draw=black, thick, anchor=center] at (1.0,1.0) {$E$};
    \end{pgfonlayer}
    \draw [fill=\fillA] (0.5,3.3)
        to (0.5,1.0)
        to (1.5,1.0) to (1.5,2 |- M.north)
        to [out=up, in=down] (2.5,2.0 |- U.south) to (2.5,3.3);
    \draw (0.5,-0.5) to (0.5,1.5);
    \draw (1.5,1.6 |- M.south)
        to [out=down, in=down, looseness=2] (2.5,0.5 |- M.south) to [out=up, in=down] (3.5,2 |- U.south)
        to (3.5,3.3);
    \draw [dashed] (2.0,-1) to (2.0,3.3);
    \node at (1,-0.85) {Alice};
    \node at (3,-0.85) {Bob};
\end{tikzpicture}
\end{aligned}
\quad=\quad
\begin{aligned}
\begin{tikzpicture}[scale=\diagramscale, thick]
    \begin{pgfonlayer}{foreground}
    \end{pgfonlayer}
    \draw [fill=\fillA] (0.5,3.3)
        to (0.5,1.5)
        to [out=down, in=down, looseness=2] (1.5,2 |- M.north)
        to [out=up, in=down] (2.5,2.0 |- U.south) to (2.5,3.3);
    \draw (0.5,-0.5) to [out=up, in=down] (3.5,1.2) to (3.5,3.3);
    \draw [dashed] (2.0,-1) to (2.0,3.3);
    \node at (1,-0.85) {Alice};
    \node at (3,-0.85) {Bob};
\end{tikzpicture}
\end{aligned}
\end{equation}
The left-hand side of this equation describes the encrypted communication protocol itself, while the right-hand side describes its intended effect. Equating the two represents the assertion that the protocol is correctly implemented. The dashed vertical line, which is not part of the formalism, represents the separation of ownership between Alice and Bob which is of importance to our interpretation.

In these diagrams, regions represent public information, lines represent computational systems, and vertices represent computational processes. In the example above, $E$ represents encryption, a process that consumes private data and publishes it as public data, while $D$ represents a decryption process, which modifies private data in a way that depends on the public data.  Note that this approach differs from the one taken by the theory of information flow~\cite{w05-cae}, where every level of security is a 0-cell.

A key advantage of our scheme is that the interpretation of a computational process depends entirely on its type, which here refers not only to its domain and codomain, but also to the entire local configuration around the vertex in a 2\-dimensional sense. Rules governing the interaction between private and public data are enforced automatically by the formalism, such that impossible or absurd operations --- such as a local system modifying nonlocally-held public data, or making use of data to which it does not have access --- cannot even be expressed. This is a strong form of locality, which is a natural and automatic property of the bicategorical formalism.

Remarkably, the form of the graphical equation~\eqref{eq:introencryptedcommunication} corresponds exactly to that of the equation for \textit{quantum teleportation}, as described in the bicategorical approach to quantum information~\cite{v12-hsqp}. One of the most important procedures in quantum theory, and yet uncovered only relatively recently~\cite{b93-teleport}, quantum teleportation is a procedure whereby two parties who share pre-existing quantum entanglement can transmit a quantum state between them, by only communicating classical information. A strong comparison to classical encrypted communication can be made: two parties who share a pre-agreed secret key can transmit a secret message between them, by only communicating public information. While easily drawn, this analogy between quantum teleportation and classical encrypted communication does not to our knowledge appear in the literature.

Using our bicategorical formalism we are able to take this comparison seriously, developing an abstract categorical description of encryption that makes the analogy mathematically precise. This indicates a close link between quantum and classical information which has not previously been recognized. We can loosely describe this correspondence in the following way:

\vspace{5pt}
\hspace{2pt}
\begin{tabular}{ll}
\bf Classical & \bf Quantum
\\
Private information & Quantum information
\\
Public information & Classical information
\\
Publication & Measurement
\\
One-time pad creation & Entangled state creation
\end{tabular}

\vspace{5pt}
\noindent
Just as the one-time pad is a fundamental resource for encrypted communication, so quantum entanglement is a fundamental resource for quantum teleportation. This paper demonstrates that the relationship is not merely analogous, but mathematically exact, with quantum randomness and classical nondeterminism giving rise to the same  formal structures.

To implement these classical protocols we must choose a bicategory in which to apply our higher syntax. We show that for classical nondeterministic computation, the symmetric monoidal bicategory \cat{2Rel} of matrices of relations provides the correct higher algebraic setting, which we define in detail in Section~\ref{sec:2rel}. Relations provide a standard semantics for nondeterministic computation~\cite{p09-qcs}, and our bicategory builds on this. Solutions to our graphical equations in this bicategory correspond to actual implementation schemes for the protocols in a classical nondeterministic setting. Some degree of nondeterminism is essential; for example, creation of a one-time pad would not be cryptographically useful if the same secret key was created every time.

Having introduced our bicategory \cat{2Rel}, we describe our abstract bicategorical syntax in Sections~\ref{sec:privateinformation} and~\ref{sec:publicinformation}. We apply this to encrypted communication, secret sharing and key exchange procedures in Section~\ref{sec:classicalprocedures}.

\section{A Bicategory of matrices of relations}
\label{sec:2rel}

\subsection{Construction}

\noindent
We now describe the bicategory \cat{2Rel} which will be the target for our constructions. It can be described quite simply in terms of finite sets and partitions: 0-cells are finite sets, 1\-cells are finite sets partitioned by their source and target sets, and 2\-cells are relations getting along with the partitioning. All the structure of a bicategory can be defined quite naturally here. We give a careful definition below, although for must purposes an intuitive understanding of the structure is quite adequate.

The $n$-cells of \cat{2Rel} are defined in the following way. \textbf{0\-cells}  are finite sets, denoted $S, T, \ldots$. A \textbf{1-cell} $A: S \to T$ is a family of finite sets $A_{t,s}$ indexed by $s \in S$ and $t \in T$. For 1-cells $A,B: S \to T$, a \textbf{2-cell} $\rho : A \Rightarrow B$ is a family of relations $\rho_{t,s}: A_{t,s} \to B_{t,s}$ indexed by $s \in S$ and $t \in T$.\ignore{More concisely, we could present a 1-cell $A:S \to T$ as a function $A \to S \times T$ for $A$ some finite set, and a 2-cell $\rho: A \Rightarrow B$ as a relation $A \to B$ making the evident triangle commute.} 

To demonstrate that these form a bicategory, we first observe that for each pair of 0\-cells $S,T$, the 1\-cells \mbox{$S \to T$} and the 2\-cells between them form a category in a straightforward way, using ordinary relational composition. Identity 1\-cells $\id_S : S \to S$ are chosen as the family $\delta _{s,s'}$, which is defined as the 1\-element set if $s=s'$ and the 0\-element set otherwise. Horizontal composition is a family of functors
\begin{equation}
\circ: \Hom(S,T) \times \Hom(T,U) \to \Hom(S,U)
\end{equation}
for each ordered triple $S,T,U$ of 0\-cells. On 1\-cells \mbox{$A:S \to T$} and \mbox{$B: T \to U$}, we define this as
\begin{equation}
\label{eq:horizontalcomposition}
(B \circ A) _{u,s} = \coprod _{t \in T} B_{u,t} \times A_{t,s}.
\end{equation}
This extends to 2\-cells in a natural way.

The final pieces of structure are the structural 2\-cells of the bicategory. For each family of composable 1\-cells $A:S \to T$, $B:T \to U$ and $C:U \to V$ we require an invertible 2\-cell
\begin{align}
\phi _{A,B,C} &: (C \circ B) \circ A \Rightarrow C \circ (B \circ A).
\end{align}
Writing out the source and target using definition~\eqref{eq:horizontalcomposition}, we define $\phi$ as the composite of canonical isomorphisms
\begin{align}
\nonumber
&\textstyle \coprod_{t} \big( ( \coprod _u C_{v,u} \times B _{u,t} ) \times A_{t,s} \big) 
\\
\nonumber
\simeq{} &
\textstyle \coprod_{t} \coprod _u \big( ( C_{v,u} \times B _{u,t} ) \times A_{t,s} \big) 
\\
\nonumber
\simeq{} &
\textstyle \coprod_{u} \coprod _t \big( C_{v,u} \times ( B _{u,t} \times A_{t,s} ) \big) 
\\
\simeq {}&
\textstyle
\coprod_{u} \big( C_{v,u} \times ( \coprod _t B _{u,t} \times A_{t,s} ) \big).
\end{align}
For each 1-cell $A:S \to T$ we also require invertible unit 2\-cells
\begin{align}
\lambda _{A} &: I_T \circ A \to A,
\\
\rho _A &: A \circ I_S \to A.
\end{align}
We define $\lambda_A$ and $\rho_A$ as the obvious isomorphisms
\begin{align}
\textstyle \coprod _{t'} (\id_T)_{t,t'} \times A_{t',s}
&=
\textstyle \coprod _{t'} \delta _{t,t'} \times A_{t',s} \simeq A_{t,s}
\\
\textstyle \coprod _{t'} A _{t,s'} \times (\id_S)_{s',s}
&=
\textstyle \coprod _{t'} A _{t,s'} \times \delta _{s',s} \simeq A_{t,s}
\end{align}
It is then straightforward to show that the required pentagon and triangle equations commute.

The bicategory \cat{2Rel} also has the following property for endomorphisms.
\begin{lemma}
\label{lem:endoinvertible}
In \cat{2Rel}, if 2\-cells $\sigma$ and $\tau$ are endomorphisms, then $\sigma \circ \tau = \id$ implies $\tau \circ \sigma = \id$.
\end{lemma}
\begin{proof}
Suppose at first that $\sigma$ and $\tau$ are relations on a finite set $S$. Then if $\sigma \circ \tau = \id_S$, there must be at least  one $y \in S$ such that $(x,y) \in \sigma$ and $(y,x) \in \tau$. But then there must be exactly one such $y$, otherwise we could not ensure that $x \neq z \in S$ implies $\not\exists y \in S$ with $(x,y) \in \sigma$ and $(y,x) \in \tau$. It follows that $\sigma$ and $\tau$ are graphs of mutually inverse bijections, and so in particular $\tau \circ \sigma = \id_S$ also.

We now turn to the general case, for which $\sigma,\tau:A \Rightarrow A$ are 2\-cells on some $A: S \to T$. But then $\sigma$ and $\tau$ are defined to be a family of relations $\sigma _{t,s}$ and $\tau_{t,s}$, and the condition $\sigma \circ \tau = \id _A$ reduces to the condition that for all $s\in S$ and $t \in T$, $\sigma _{t,s} \circ \tau_{t,s} = \id _{A_{t,s}}$. By the argument above this implies that $\tau _{t,s} \circ \sigma _{t,s} = \id _{A _{t,s}}$, and hence $\tau \circ \sigma = \id _A$.
\end{proof}

\subsection{Symmetric monoidal structure}

\noindent
In fact, \cat{2Rel} can be given the structure of a symmetric monoidal bicategory, for which the tensor product of two 0\-cells is their cartesian product as sets. For full details see~\cite{s13-ccb}, in which an equivalent bicategory~\cat{Mat(Rel)} is described.  Here, 0\-cells correspond to finite cardinalities, 1\-cells correspond to matrices of sets, and 2\-cells correspond to matrices of relations.  The monoidal structure is the usual tensor product of matrices, also known as the Kr\"onecker product.  The tensor product of an $m\times n$ matrix with an $r\times s$ matrix is an $mr \times ns$ matrix. 

The monoidal unit for this product is the 1-element set in \cat{2Rel}.  This labels the empty region in the graphical calculus.  We can then construct the \emph{scalars}, defined as the category $\Hom(1,1)$, represented in the graphical calculus as lines and boxes on a white background. The scalars of a symmetric monoidal bicategory necessarily form a symmetric monoidal category, which in our case is simply \cat{Rel}, the symmetric monoidal category of finite sets and relations.

In our formalism, regions are labelled by types of public information.  No information is needed to pick the single element of the one-element set, so restricting attention to the scalars implies neglecting all nontrivial public information.  What remains is private computational systems and their dynamics, and so we see that \cat{2Rel} treats purely private computations as arbitrary nondeterministic processes.

\section{Private information}
\label{sec:privateinformation}

\subsection{String diagrams}
We assume that a single, isolated computational system is located at any moment at a single point in space, and so over time its history traces out a line in spacetime:
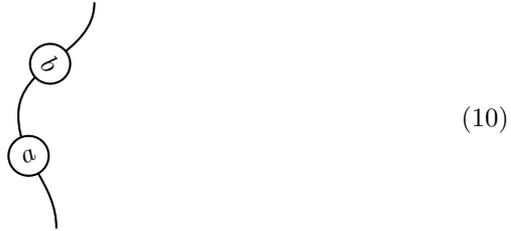
\begin{equation}
\label{eq:privateinformation}
\begin{aligned}
\begin{tikzpicture}[
    decoration={
        markings,
        mark=at position 0.3 with
            {\node [circle, transform shape, fill=white, draw, inner sep=0pt, minimum width=15pt, rotate=-90] {$a$};},
        mark=at position 0.7 with
            {\node [circle, transform shape, fill=white, draw, inner sep=0pt, minimum width=15pt, rotate=-90] {$b$};},
    }, thick
]
\draw [postaction={decorate}] (0,0) to [out=up, in=down] (-0.5,1.5) to [out=up, in=down] (0.5,3);
\end{tikzpicture}
\end{aligned}
\end{equation}
The vertices $a$ and $b$ represent arbitrary computations that act on the system. We could have many such systems, interacting in a complicated way:
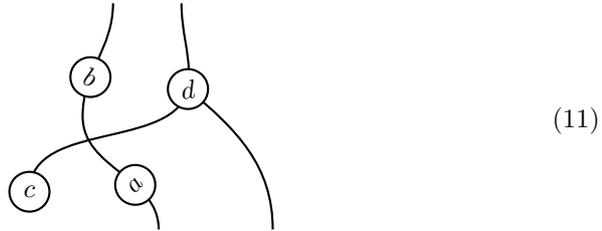
\begin{equation}
\begin{aligned}
\begin{tikzpicture}[thick]
\draw [decoration={
        markings,
        mark=at position 0.2 with
            {\node [circle, transform shape, fill=white, draw, inner sep=0pt, minimum width=15pt, rotate=-90] {$a$};},
        mark=at position 0.7 with
            {\node [circle, transform shape, fill=white, draw, inner sep=0pt, minimum width=15pt, rotate=-90] {$b$};},
    }, postaction={decorate}, thick] (0.5,0) to [out=up, in=down] (-0.5,1.5) to [out=up, in=down] (-0.1,3);
\begin{pgfonlayer}{foreground}
\draw [decoration={
        markings,
        mark=at position 0.0 with
            {\node [circle, transform shape, fill=white, draw, inner sep=0pt, minimum width=15pt, rotate=-90] {$c$};},
        mark=at position 0.7 with
            {\node (d) [circle, transform shape, fill=white, draw, inner sep=0pt, minimum width=15pt, rotate=-70] {$d$};},
    }, postaction={decorate}, thick] (-1.2,0.5) to [out=up, in=down] (0.9,2) to [out=up, in=down] (0.8,3);
\end{pgfonlayer}
\draw (2,0.0) to [out=up, in=-40] (d.center);
\end{tikzpicture}
\end{aligned}
\end{equation}
This diagram describes two pre-existing systems, and a third system which is produced from a computational process $c$ with no input. Two of the systems switch positions without interacting, represented by the crossed worldlines. A process $d$ then takes place, which takes two systems as input and produces one system as output.

These diagrams have already found extensive use in the foundations of computer science and logic~\cite{bs10-rosetta}, and also in the foundations of quantum computing~\cite{ac08-cqm}. They are often called \emph{string diagrams}, and are a rigorous and powerful notation for morphisms in symmetric monoidal categories~\cite{s11-sgl}. Strings correspond to objects of the monoidal category, vertices correspond to morphisms, and placing diagrams side-by-side corresponds to the tensor product operation.

We assume that our string diagrams are valued in \cat{Rel}, the symmetric monoidal category of finite sets and relations. This forms the scalars of \cat{2Rel}, as discussed in Section~\ref{sec:2rel}. We will interpret an object of \cat{Rel} as representing a classical computational system, with a particular finite set of internal states. Morphisms are interpreted as computational dynamics, nondeterministically transforming states of the domain into states of the codomain.

\subsection{Self-dualizability and one-time pads}
A system is called \emph{self-dualizable} if it can be equipped with unit and counit morphisms
\begin{calign}
\label{eq:unitcounit}
\begin{aligned}
\begin{tikzpicture}[yscale=-1, thick]
\draw [white] (0.5,0) to (0.5,1);
\draw (0,0) to [out=up, in=up, looseness=2] (1,0);
\end{tikzpicture}
\end{aligned}
&
\begin{aligned}
\begin{tikzpicture}[thick]
\draw [white] (0.5,0) to (0.5,1);
\draw (0,0) to [out=up, in=up, looseness=2] (1,0);
\end{tikzpicture}
\end{aligned}
\end{calign}
satisfying the following equations, called the snake equations:
\begin{equation}
\label{eq:snakeequations}
\begin{aligned}
\begin{tikzpicture}[scale=\diagramscale, thick]
\draw (0,-0.5) to (0,0.5) to [out=up, in=up, looseness=2] (1,0.5) to [out=down, in=down, looseness=2] (2,0.5) to (2,1.5);
\end{tikzpicture}
\end{aligned}
\quad=\quad
\begin{aligned}
\begin{tikzpicture}[scale=\diagramscale, thick]
\draw (0,-0.5) to (0,1.5);
\end{tikzpicture}
\end{aligned}
\quad=\quad
\begin{aligned}
\begin{tikzpicture}[xscale=-1, scale=\diagramscale, thick]
\draw (0,-0.5) to (0,0.5) to [out=up, in=up, looseness=2] (1,0.5) to [out=down, in=down, looseness=2] (2,0.5) to (2,1.5);
\end{tikzpicture}
\end{aligned}
\end{equation}
We say that the unit and counit morphisms \textit{witness} the self-duality. In \cat{FRel} every object $A$ is self-dualizable, with the unit morphism $\eta: 1 \to A \times A$ given canonically by \mbox{$\eta = \sum_{a \in A} (a,a)$}, and with the counit given by the converse of this relation.

Not every unit and counit map witnessing self-dualizability will be of this form, but they can be characterized in the following way.

\begin{lemma}
In a monoidal category, for a self-dualizable object $A$, there is a bijection between choices of unit and counit morphism, and isomorphisms $A \simeq A$.
\end{lemma}
\begin{proof}
Since $A$ is self-dualizable, we can pick unit and counit morphisms~\eqref{eq:unitcounit} witnessing this. Given a second unit and counit
\begin{calign}
\label{eq:secondunitcounit}
\begin{aligned}
\begin{tikzpicture}[yscale=-1, scale=\diagramscale, thick]
\draw [white] (0.5,0) to (0.5,1);
\draw (0,-0.2) to [] (0,0.5);
\draw (1,-0.2) to [] (1,0.5);
\node [draw, minimum width=1.2cm, minimum height=0.3cm, fill=white, anchor=north] at (0.5,0.5) {};
\end{tikzpicture}
\end{aligned}
&\begin{aligned}
\begin{tikzpicture}[scale=\diagramscale, thick]
\draw [white] (0.5,0) to (0.5,1);
\draw (0,-0.2) to [] (0,0.5);
\draw (1,-0.2) to [] (1,0.5);
\node [draw, minimum width=1.2cm, minimum height=0.3cm, fill=white, anchor=south] at (0.5,0.5) {};
\end{tikzpicture}
\end{aligned}
\end{calign}
also witnessing a self-duality, we can construct the following morphisms of type $A \to A$:
\begin{calign}
\begin{aligned}
\begin{tikzpicture}[scale=\diagramscale, thick]
\draw (0,-0.5) to (0,0.5) to [out=up, in=up, looseness=2] (-1,0.5);
\draw (-2,0.5) to (-2,1.5);
\node [draw, minimum width=1.3cm, minimum height=0.3cm, fill=white, anchor=north] at (-1.5,0.5) {};
\end{tikzpicture}
\end{aligned}
&
\begin{aligned}
\begin{tikzpicture}[scale=\diagramscale, xscale=-1, yscale=-1, thick]
\draw (0,-0.5) to (0,0.5) to [out=up, in=up, looseness=2] (-1,0.5);
\draw (-2,0.5) to (-2,1.5);
\node [draw, minimum width=1.3cm, minimum height=0.3cm, fill=white, anchor=south] at (-1.5,0.5) {};
\end{tikzpicture}
\end{aligned}
\end{calign}
Applying the snake equations~\eqref{eq:snakeequations} it can be shown that these morphisms are inverse to each other. Conversely, given an isomorphism $s:A \to A$, we can form the following unit and counit morphisms:
\begin{calign}
\begin{aligned}
\begin{tikzpicture}[yscale=-1, thick]
\draw [white] (0.5,0) to (0.5,1);
\draw (0,-0.7) to (0,0) node [anchor=south, draw, fill=white, inner sep=0pt, minimum width=17pt, minimum height=12pt] {$s$} to [out=up, in=up, looseness=2] (1,0) to (1,-0.7);
\end{tikzpicture}
\end{aligned}
&
\begin{aligned}
\begin{tikzpicture}[xscale=-1, thick]
\draw [white] (0.5,0) to (0.5,1);
\draw (0,-0.7) to (0,0) node [anchor=north, draw, fill=white, inner sep=0pt, minimum width=17pt, minimum height=12pt] {$s ^{-1}$} to [out=up, in=up, looseness=2] (1,0) to (1,-0.7);
\end{tikzpicture}
\end{aligned}
\end{calign}
It is straightforward to show that these constructions are inverse, so we have a bijection between unit and counit morphisms and automorphisms, as desired.
\end{proof}

In \cat{Rel}, the automorphisms of an object are exactly the bijections. As a result every unit morphism \mbox{$\eta: I \to S \times S$} is of the form $\sum_s (s,\pi(s))$ for some permutation $\pi$ of $S$. That is, the unit morphisms represent nondeterministic processes whereby the first party receives an arbitrary $s \in S$, and the second party receives $\pi(s)$. If the permutation $\pi$ is known, its inverse can be applied by the second party, and both parties will then share matching keys which can be used as a cryptographic resource. So given a self-dualizable object, we can interpret a unit morphism as a \textit{key exchange procedure}. The counit can similarly be interpreted as a \emph{key verification} procedure, which terminates the computation iff the two parties have mismatched keys.

\subsection{Kernels, deletion and random data }
Morphisms of \cat{Rel} can have elements of their domain which are not related to any elements of their codomain. These describe situations where the computation halts. Given a relation \mbox{$\rho:A \to B$}, its \emph{kernel} is a relation $\kappa: K \to A$ such that $\rho \circ \kappa = 0$, the empty relation, and such that $\kappa$ is universal with this property:
\begin{equation}
\begin{aligned}
\begin{tikzpicture}[scale=2]
\node (K) at (0,0) {$K$};
\node (A) at (1,0) {$A$};
\node (B) at (2,0) {$B$};
\node (X) at (0.5,-0.5) {$X$};
\draw [->] (K) to node [above] {$\kappa$} (A);
\draw [->] (X) to node [auto, swap] {$\sigma$} (A);
\draw [->, dashed] (X) to node [auto] {$\tilde \sigma$} (K);
\draw [->] (A) to [out=10, in=170] node [above] {$\rho$} (B);
\draw [->] (A) to [out=-10, in=-170] node [below] {$0$} (B);
\end{tikzpicture}
\end{aligned}
\end{equation}
The universal property is that for all relations $\sigma: X \to A$ with $\rho \circ \sigma = 0$, then $\sigma$ factors through $\kappa$. The morphism $\kappa$ then characterizes the elements of $A$ on which $\rho$ halts. The construction of kernels extends in a similar way to arbitrary 2\-cells in \cat{2Rel}.

For a finite set $A$ there is a unique relation of type $A \to 1$ that has zero kernel. We interpret this as a process that eliminates the system $A$, without halting the computation. We denote this graphically in the following way:
\begin{equation}
\begin{aligned}
\begin{tikzpicture}[thick]
\draw (0,0) to (0,1) node [Vertex] {};
\end{tikzpicture}
\end{aligned}
\end{equation}
The converse process represents the nondeterministic preparation of a system in an arbitrary, `random'  state:
\begin{equation}
\begin{aligned}
\begin{tikzpicture}[yscale=-1, thick]
\draw (0,0) to (0,1) node [Vertex] {};
\end{tikzpicture}
\end{aligned}
\end{equation}
These are related by the unit and counit morphisms~\eqref{eq:unitcounit} witnessing self-dualizability via the following equations:
\begin{calign}
\label{eq:twistdelete1}
\begin{aligned}
\begin{tikzpicture}[scale=\diagramscale, thick]
\draw (0,0) node [Vertex] {} to [out=down, in=left, looseness=1.2] (0.5,-0.6) to [out=right, in=down, looseness=1.2] (1,0) to (1,1);
\end{tikzpicture}
\end{aligned}
\quad=\quad
\begin{aligned}
\begin{tikzpicture}[scale=\diagramscale, thick]
\draw [white] (-0.05,-0.6) to +(0.1,0);
\draw (0,0) node [Vertex] {} to (0,1);
\end{tikzpicture}
\end{aligned}
\quad=\quad
\begin{aligned}
\begin{tikzpicture}[scale=\diagramscale, xscale=-1, thick]
\draw (0,0) node [Vertex] {} to [out=down, in=left, looseness=1.2] (0.5,-0.6) to [out=right, in=down, looseness=1.2] (1,0) to (1,1);
\end{tikzpicture}
\end{aligned}
\\[10pt]
\label{eq:twistdelete2}
\begin{aligned}
\begin{tikzpicture}[scale=\diagramscale, yscale=-1, thick]
\draw (0,0) node [Vertex] {} to [out=down, in=left, looseness=1.2] (0.5,-0.6) to [out=right, in=down, looseness=1.2] (1,0) to (1,1);
\end{tikzpicture}
\end{aligned}
\quad=\quad
\begin{aligned}
\begin{tikzpicture}[scale=\diagramscale, yscale=-1, thick]
\draw [white] (-0.05,-0.6) to +(0.1,0);
\draw (0,0) node [Vertex] {} to (0,1);
\end{tikzpicture}
\end{aligned}
\quad=\quad
\begin{aligned}
\begin{tikzpicture}[scale=\diagramscale, xscale=-1, yscale=-1, thick]
\draw (0,0) node [Vertex] {} to [out=down, in=left, looseness=1.2] (0.5,-0.6) to [out=right, in=down, looseness=1.2] (1,0) to (1,1);
\end{tikzpicture}
\end{aligned}
\end{calign}
Each of these has a natural interpretation in terms of nondeterministic classical computation: the equalities~\eqref{eq:twistdelete1} say that if you nondeterministically create shared keys and then delete one of the keys, the remaining key is uniformly random; while the equalities~\eqref{eq:twistdelete2} say that if you have a given key, it is always possible that another key produced nondeterministically might match it.

\section{Public information}
\label{sec:publicinformation}

\subsection{Graphical calculus}
We now consider a  graphical notation for correlation between many computational systems. Already explored in the context of quantum information~\cite{v12-hsqp}, here  we investigate its applications to classical information processing for the first time. Consider a family of systems carrying private data, existing simultaneously  without interacting. We can draw this straightforwardly in our string diagram notation as follows:
\def\amplitude{0.3}
\begin{equation}
\begin{aligned}
\begin{tikzpicture}
\pgfmathsetseed{1}
\foreach \x in {0.0,0.1,...,1.0}
{
    \pgfmathsetmacro{\inrand}{1+(rnd-0.5)*\amplitude}
    \pgfmathsetmacro{\outrand}{1+(rnd-0.5)*\amplitude}
    \draw (\x,0) to [out=up, in=down, in looseness=\inrand, out looseness=\outrand]
        (\x+0.4,2);
}
\end{tikzpicture}
\end{aligned}
\end{equation}
Each vertical line represents a separate computational system.

Now suppose that all of these systems hold the same information, in a completely redundant way. Inventing a new notation, we indicate this by shading the effective 2\-dimensional area swept out by the worldlines of our systems:
\begin{equation}
\begin{aligned}
\begin{tikzpicture}[]
\draw [fill=\fillA, draw=none] (0,0) to [out=up, in=down] (0.4,2) to (1.4,2) to [out=down, in=up] (1,0) to (0,0);
\pgfmathsetseed{5}
\foreach \x in {0.1,0.2,...,0.9}
{
    \pgfmathsetmacro{\inrand}{1+(rnd-0.5)*\amplitude}
    \pgfmathsetmacro{\outrand}{1+(rnd-0.5)*\amplitude}
    \draw (\x,0) to [out=up, in=down, in looseness=\inrand, out looseness=\outrand]
        (\x+0.4,2);
}
\draw (0,0) to [out=up, in=down]
    (0+0.4,2);
\draw (1,0) to [out=up, in=down]
    (1+0.4,2);
\end{tikzpicture}
\end{aligned}
\qquad\leadsto\qquad
\begin{aligned}
\begin{tikzpicture}[, thick]
\draw [fill=\fillA, draw=none] (0,0) to [out=up, in=down] (0.4,2) to (1.4,2) to [out=down, in=up] (1,0) to (0,0);
\foreach \x/\y in {0.0/1, 1.0/1}
{
\draw (\x,0) to [out=up, in=down, looseness=\y] (\x+0.4,2);
}
\end{tikzpicture}
\end{aligned}
\end{equation}
We have presented this as nothing more than a notational convenience. But in fact, if we include these regions formally as elements of our notation, we obtain precisely the graphical notation for a bicategory. So our richer formalism has a rigorous mathematical foundation, extending that of our original notation.

We interpret these 2-dimensional areas as representing \textit{public information}, contrasting with the 1-dimensional lines in Section~\ref{sec:privateinformation} representing private information. Private information is held at a single point in space, and can be controlled or manipulated however its owner desires. Public information can be accessed at any point on its worldsheet, but cannot be modified by local actions, since it is held redundantly over a finite spatial region. So public information is more accessible, but as a consequence less mutable.

This can be considered an abstraction of real public information storage systems, such as the Domain Name Service, which stores public information redundantly on many independent computers. This makes the data easier to access, since it is more likely there will be a copy of the data nearby that can be consulted. But the downside is that information update is no longer a local operation: complex algorithms are required to synchronize the information held by the individual computers. It would be interesting to consider whether an extension of our formalism could address these issues of distributed computation.

Since we are thinking intuitively of  public information as formed from a large collection of correlated systems, it makes sense that we should be able to copy the public information by splitting this family of systems into two parts, and delete the information by deleting each constituent system. We denote these operations in the following way:
\def\boxwidth{2cm}
\begin{align}
\label{eq:copypublicdata}
\makebox[\boxwidth]{$\begin{aligned}
\begin{tikzpicture}
\pgfmathsetseed{1}
\draw [fill=\fillA, draw=none] (-1.05,-0.0) rectangle (0.95,1.8);
\draw [fill=white, draw=none] (-0.55,1.8) to (-0.55,1.5) to [out=down, in=down, looseness=2] (0.45,1.5) to (0.45,1.8);
\foreach \x in {0.0,0.1,...,1.0}
{
    \pgfmathsetmacro{\inrand}{1+(rnd-0.5)*\amplitude}
    \pgfmathsetmacro{\outrand}{2.0+(rnd-0.5)*\amplitude}
    \draw (\x,0.0) to [out=up, in=down, in looseness=\inrand, out looseness=\outrand]
        (0.45+0.5*\x,1.5) to +(0,0.3);
    \pgfmathsetmacro{\inrand}{1+(rnd-0.5)*\amplitude}
    \pgfmathsetmacro{\outrand}{2.0+(rnd-0.5)*\amplitude}
    \draw (-\x-0.1,0) to [out=up, in=down, in looseness=\inrand, out looseness=\outrand]
        (-0.55-0.5*\x,1.5) to +(0,0.3);
}
\end{tikzpicture}
\end{aligned}$}
\qquad&\leadsto\qquad
\makebox[\boxwidth]{$\begin{aligned}
\begin{tikzpicture}[, thick]
\pgfmathsetseed{1}
\draw [fill=\fillA, draw=none] (-1.05,-0.0) rectangle (0.95,1.8);
\draw [fill=white] (-0.55,1.8) to (-0.55,1.5) to [out=down, in=down, looseness=2] (0.45,1.5) to (0.45,1.8);
\end{tikzpicture}
\end{aligned}$}
\\
\label{eq:deletepublicdata}
\setlength{\vertexradius}{0.8pt}
\makebox[\boxwidth]{$\begin{aligned}
\begin{tikzpicture}
\draw [fill=\fillA, draw=none, name path=curve] (0,0) to [out=up, in=down] (0.0,1) to [out=up, in=up, looseness=1] (1.0,1) to [out=down, in=up] (1,0);
\pgfmathsetseed{6}
\foreach \x in {0.1,0.2,...,0.9}
{
    \pgfmathsetmacro{\srand}{rand*0.02}
    \draw [draw=none, name path=straight] (\x,0) to [out=up, in=down]
        (\x+\srand,1.6);
    \draw [name intersections={of=curve and straight}] (\x,0) to (intersection-1) node [Vertex] {};
}
\draw (0,0) to [out=up, in=down]
    (0+0.0,1) node [Vertex] {};
\draw (1,0) to [out=up, in=down]
    (1+0.0,1) node [Vertex] {};
\end{tikzpicture}
\end{aligned}$}
\qquad&\leadsto\qquad
\makebox[\boxwidth]{$\begin{aligned}
\begin{tikzpicture}[, thick]
\draw [fill=\fillA, name path=curve] (0,0) to [out=up, in=down] (0.0,1) to [out=up, in=up, looseness=1.4] (1.0,1) to [out=down, in=up] (1,0);
\end{tikzpicture}
\end{aligned}$}
\end{align}
On the left-hand side is the intuitive picture in terms of families of perfectly correlated computational systems, and on the right-hand side is the formal component of our graphical calculus which represents it. We can also form the converses of these operations:
\begin{align}
\label{eq:comparepublicdata}
\makebox[\boxwidth]{$\begin{aligned}
\begin{tikzpicture}[yscale=-1]
\pgfmathsetseed{1}
\draw [fill=\fillA, draw=none] (-1.05,-0.0) rectangle (0.95,1.8);
\draw [fill=white, draw=none] (-0.55,1.8) to (-0.55,1.5) to [out=down, in=down, looseness=2] (0.45,1.5) to (0.45,1.8);
\foreach \x in {0.0,0.1,...,1.0}
{
    \pgfmathsetmacro{\inrand}{1+(rnd-0.5)*\amplitude}
    \pgfmathsetmacro{\outrand}{2.0+(rnd-0.5)*\amplitude}
    \draw (\x,0.0) to [out=up, in=down, in looseness=\inrand, out looseness=\outrand]
        (0.45+0.5*\x,1.5) to +(0,0.3);
    \pgfmathsetmacro{\inrand}{1+(rnd-0.5)*\amplitude}
    \pgfmathsetmacro{\outrand}{2.0+(rnd-0.5)*\amplitude}
    \draw (-\x-0.1,0) to [out=up, in=down, in looseness=\inrand, out looseness=\outrand]
        (-0.55-0.5*\x,1.5) to +(0,0.3);
}
\end{tikzpicture}
\end{aligned}$}
\qquad&\leadsto\qquad
\makebox[\boxwidth]{$\begin{aligned}
\begin{tikzpicture}[yscale=-1, thick]
\pgfmathsetseed{1}
\draw [fill=\fillA, draw=none] (-1.05,-0.0) rectangle (0.95,1.8);
\draw [fill=white] (-0.55,1.8) to (-0.55,1.5) to [out=down, in=down, looseness=2] (0.45,1.5) to (0.45,1.8);
\end{tikzpicture}
\end{aligned}$}
\\[5pt]
\label{eq:createpublicdata}
\setlength{\vertexradius}{0.8pt}
\makebox[\boxwidth]{$\begin{aligned}
\begin{tikzpicture}[yscale=-1]
\draw [fill=\fillA, draw=none, name path=curve] (0,0) to [out=up, in=down] (0.0,1) to [out=up, in=up, looseness=1] (1.0,1) to [out=down, in=up] (1,0);
\pgfmathsetseed{6}
\foreach \x in {0.1,0.2,...,0.9}
{
    \pgfmathsetmacro{\srand}{rand*0.02}
    \draw [draw=none, name path=straight] (\x,0) to [out=up, in=down]
        (\x+\srand,1.6);
    \draw [name intersections={of=curve and straight}] (\x,0) to (intersection-1) node [Vertex] {};
}
\draw (0,0) to [out=up, in=down]
    (0+0.0,1) node [Vertex] {};
\draw (1,0) to [out=up, in=down]
    (1+0.0,1) node [Vertex] {};
\end{tikzpicture}
\end{aligned}$}
\qquad&\leadsto\qquad
\makebox[\boxwidth]{$\begin{aligned}
\begin{tikzpicture}[yscale=-1, thick]
\draw [fill=\fillA, name path=curve] (0,0) to [out=up, in=down] (0.0,1) to [out=up, in=up, looseness=1.4] (1.0,1) to [out=down, in=up] (1,0);
\end{tikzpicture}
\end{aligned}$}
\end{align}
The first of these represents the process of comparing two pieces of public data. In the case that the values are different, this cannot be successful and we might expect the computation to halt, which will be demonstrated  by the concrete relational model we examine below. The second represents the creation of public data in a nondeterministic uniform fashion.

\subsection{Topological axioms}
As with the bicategorical syntax for quantum information~\cite{v12-hsqp}, in order to support their interpretations, we require these copying, deleting, comparison and uniform creation components to satisfy certain equations. They are topological, in that they amount to saying that any composite diagram is determined only by its connectivity.
{\def\diagramscale{0.5}
\begin{align}
\label{eq:topological1}
\begin{aligned}
\begin{tikzpicture}[scale=\diagramscale, thick]
\draw [use as bounding box, draw=none] (-0.5,0) rectangle (2.3,-2);
\draw [white] (-0.5,0) to (3.1,-2);
\draw [fill=\fillA, draw=none] (-0.5,0) to (0.3,0) to (0.3,-1)
    to [out=down, in=down, looseness=1.5] (1.3,-1)
    to [out=up, in=up, looseness=1.5] (2.3,-1)
    to (2.3,-2) to (-0.5,-2);
\draw [] (0.3,0) to (0.3,-1)
    to [out=down, in=down, looseness=1.5] (1.3,-1)
    to [out=up, in=up, looseness=1.5] (2.3,-1)
    to (2.3,-2);
\end{tikzpicture}
\end{aligned}
\hspace{5pt}&=\hspace{5pt}
\begin{aligned}
\begin{tikzpicture}[scale=\diagramscale, thick]
\draw [use as bounding box, draw=none] (1,0) rectangle (0,2);
\draw [white] (0,0) to (2,2);
\draw [fill=\fillA, draw=none] (0,0)
    to (1,0)
    to (1,2)
    to (0,2);
\draw [] (1,0) to (1,2);
\end{tikzpicture}
\end{aligned}
&
\begin{aligned}
\begin{tikzpicture}[scale=\diagramscale, thick]
\draw [use as bounding box, draw=none] (-2.3,0) rectangle (0.5,2);
\draw [white] (0.5,0) to (-3.1,2);
\draw [fill=\fillA, draw=none] (0.5,0) to (-0.3,0) to (-0.3,1)
    to [out=up, in=up, looseness=1.5] (-1.3,1)
    to [out=down, in=down, looseness=1.5] (-2.3,1)
    to (-2.3,2) to (0.5,2);
\draw [] (-0.3,0) to (-0.3,1)
    to [out=up, in=up, looseness=1.5] (-1.3,1)
    to [out=down, in=down, looseness=1.5] (-2.3,1)
    to (-2.3,2);
\end{tikzpicture}
\end{aligned}
\hspace{5pt}&=\hspace{5pt}
\begin{aligned}
\begin{tikzpicture}[scale=\diagramscale, thick]
\draw [use as bounding box, draw=none] (-1,0) rectangle (0,2);
\draw [white] (0,0) to (-2,2);
\draw [fill=\fillA, draw=none] (0,0)
    to (-1,0)
    to (-1,2)
    to (-0,2);
\draw [] (-1,0) to (-1,2);
\end{tikzpicture}
\end{aligned}
\end{align}
\begin{align}
\label{eq:topological2}
\begin{aligned}
\begin{tikzpicture}[scale=\diagramscale]
\draw [use as bounding box, draw=none] (-0.5,0) rectangle (2.3,2);
\draw [white] (-0.5,0) to (3.1,2);
\draw [fill=\fillA, draw=none] (-0.5,0) to (0.3,0) to (0.3,1)
    to [out=up, in=up, looseness=1.5] (1.3,1)
    to [out=down, in=down, looseness=1.5] (2.3,1)
    to (2.3,2) to (-0.5,2);
\draw [thick] (0.3,0) to (0.3,1)
    to [out=up, in=up, looseness=1.5] (1.3,1)
    to [out=down, in=down, looseness=1.5] (2.3,1)
    to (2.3,2);
\end{tikzpicture}
\end{aligned}
\hspace{5pt}&=\hspace{5pt}
\begin{aligned}
\begin{tikzpicture}[scale=\diagramscale]
\draw [use as bounding box, draw=none] (1,0) rectangle (0,2);
\draw [white] (0,0) to (2,2);
\draw [fill=\fillA, draw=none] (0,0)
    to (1,0)
    to (1,2)
    to (0,2);
\draw [thick] (1,0) to (1,2);
\end{tikzpicture}
\end{aligned}
&
\begin{aligned}
\begin{tikzpicture}[scale=\diagramscale]
\draw [use as bounding box, draw=none] (-2.3,0) rectangle (0.5,-2);
\draw [white] (0.5,0) to (-3.1,-2);
\draw [fill=\fillA, draw=none] (0.5,0) to (-0.3,0) to (-0.3,-1)
    to [out=down, in=down, looseness=1.5] (-1.3,-1)
    to [out=up, in=up, looseness=1.5] (-2.3,-1)
    to (-2.3,-2) to (0.5,-2);
\draw [thick] (-0.3,0) to (-0.3,-1)
    to [out=down, in=down, looseness=1.5] (-1.3,-1)
    to [out=up, in=up, looseness=1.5] (-2.3,-1)
    to (-2.3,-2);
\end{tikzpicture}
\end{aligned}
\hspace{5pt}&=\hspace{5pt}
\begin{aligned}
\begin{tikzpicture}[scale=\diagramscale]
\draw [use as bounding box, draw=none] (-1,0) rectangle (0,2);
\draw [white] (0,0) to (-2,2);
\draw [fill=\fillA, draw=none] (0,0)
    to (-1,0)
    to (-1,2)
    to (0,2);
\draw [thick] (-1,0) to (-1,2);
\end{tikzpicture}
\end{aligned}
\end{align}
\def\diagramscale{0.4}
\begin{align}
\label{eq:commplanar}
\begin{aligned}
\begin{tikzpicture}[scale=\diagramscale]
\draw [fill=\fillA, draw=none] (0,0)
    to [out=up, in=down, out looseness=1.5] (2,2)
    to (3,2)
    to [out=down, in=up, out looseness=1.5] (1,0);
\draw [fill=\fillA, draw=none] (3,0)
    to [out=up, in=down, out looseness=1.5] (1,2)
    to (0,2)
    to [out=down, in=up, out looseness=1.5] (2,0);
\draw [fill=\fillA, draw=none] (0,2)
    to (1,2)
    to [out=up, in=up, looseness=1.5] (2,2)
    to (3,2)
    to [out=up, in=down] (2,4)
    to (1,4) to [out=down, in=up] (0,2);
\draw [thick] (0,0)
    to [out=up, in=down, out looseness=1.5] (2,2)
    to [out=up, in=up, looseness=1.5] (1,2)
    to [out=down, in=up, in looseness=1.5] (3,0);
\draw [thick] (1,0)
    to [out=up, in=down, in looseness=1.5] (3,2)
    to [out=up, in=down] (2,4);
\draw [thick] (2,0)
    to [out=up, in=down, in looseness=1.5] (0,2)
    to [out=up, in=down] (1,4);
\end{tikzpicture}
\end{aligned}
\hspace{5pt}&=\hspace{5pt}
\begin{aligned}
\begin{tikzpicture}[scale=\diagramscale]
\draw [fill=\fillA, draw=none] (0,0)
    to [out=up, in=down, out looseness=1.5] (0,2)
    to (1,2)
    to [out=down, in=up, out looseness=1.5] (1,0);
\draw [fill=\fillA, draw=none] (3,0)
    to [out=up, in=down, out looseness=1.5] (3,2)
    to (2,2)
    to [out=down, in=up, out looseness=1.5] (2,0);
\draw [fill=\fillA, draw=none] (0,2)
    to (1,2)
    to [out=up, in=up, looseness=1.5] (2,2)
    to (3,2)
    to [out=up, in=down] (2,4)
    to (1,4) to [out=down, in=up] (0,2);
\draw [thick] (3,0)
    to [out=up, in=down, in looseness=1.5] (3,2)
    to [out=up, in=down] (2,4);
\draw [thick] (0,0)
    to [out=up, in=down, in looseness=1.5] (0,2)
    to [out=up, in=down] (1,4);
\draw [thick] (2,0) 
    to (2,2)
    to [out=up, in=up, looseness=1.5] (1,2)
    to (1,0);
\end{tikzpicture}
\end{aligned}
\,\,&\,\,
\begin{aligned}
\begin{tikzpicture}[scale=\diagramscale, yscale=-1]
\draw [fill=\fillA, draw=none] (0,0)
    to [out=up, in=down, out looseness=1.5] (2,2)
    to (3,2)
    to [out=down, in=up, out looseness=1.5] (1,0);
\draw [fill=\fillA, draw=none] (3,0)
    to [out=up, in=down, out looseness=1.5] (1,2)
    to (0,2)
    to [out=down, in=up, out looseness=1.5] (2,0);
\draw [fill=\fillA, draw=none] (0,2)
    to (1,2)
    to [out=up, in=up, looseness=1.5] (2,2)
    to (3,2)
    to [out=up, in=down] (2,4)
    to (1,4) to [out=down, in=up] (0,2);
\draw [thick] (0,0)
    to [out=up, in=down, out looseness=1.5] (2,2)
    to [out=up, in=up, looseness=1.5] (1,2)
    to [out=down, in=up, in looseness=1.5] (3,0);
\draw [thick] (1,0)
    to [out=up, in=down, in looseness=1.5] (3,2)
    to [out=up, in=down] (2,4);
\draw [thick] (2,0)
    to [out=up, in=down, in looseness=1.5] (0,2)
    to [out=up, in=down] (1,4);
\end{tikzpicture}
\end{aligned}
\hspace{5pt}&=\hspace{5pt}
\begin{aligned}
\begin{tikzpicture}[scale=\diagramscale, yscale=-1]
\draw [fill=\fillA, draw=none] (0,0)
    to [out=up, in=down, out looseness=1.5] (0,2)
    to (1,2)
    to [out=down, in=up, out looseness=1.5] (1,0);
\draw [fill=\fillA, draw=none] (3,0)
    to [out=up, in=down, out looseness=1.5] (3,2)
    to (2,2)
    to [out=down, in=up, out looseness=1.5] (2,0);
\draw [fill=\fillA, draw=none] (0,2)
    to (1,2)
    to [out=up, in=up, looseness=1.5] (2,2)
    to (3,2)
    to [out=up, in=down] (2,4)
    to (1,4) to [out=down, in=up] (0,2);
\draw [thick] (3,0)
    to [out=up, in=down, in looseness=1.5] (3,2)
    to [out=up, in=down] (2,4);
\draw [thick] (0,0)
    to [out=up, in=down, in looseness=1.5] (0,2)
    to [out=up, in=down] (1,4);
\draw [thick] (2,0) 
    to (2,2)
    to [out=up, in=up, looseness=1.5] (1,2)
    to (1,0);
\end{tikzpicture}
\end{aligned}
\end{align}

\vspace{-8pt}
\begin{equation}
\label{eq:copycompare}
\begin{aligned}
\begin{tikzpicture}[scale=\diagramscale, scale=1.5]
\draw [fill=\fillA, draw=none, use as bounding box] (0.5,0.5) rectangle (2.5,2.5);
\draw [fill=white, thick] (1,1.5)    to [out=up, in=up, looseness=1.5] (2,1.5)    to [out=down, in=down, looseness=1.5] (1,1.5);
\end{tikzpicture}
\end{aligned}
\hspace{5pt}=\hspace{2pt}
\begin{aligned}
\begin{tikzpicture}[scale=\diagramscale, scale=1.5]
\draw [fill=\fillA, draw=none, use as bounding box] (0.5,0.5) rectangle (2.5,2.5);
\end{tikzpicture}
\end{aligned}
\end{equation}}Each of these equations is consistent with the interpretation we give to the basic components~\eqref{eq:copypublicdata}--\eqref{eq:createpublicdata}. For example, the first equality labelled~\eqref{eq:topological1} represents the fact that copying public information and then deleting the new copy results in the identity; the first equality labelled~\eqref{eq:commplanar} represents the fact that exchanging public information and then comparing gives the same result as simply comparing; and equation~\eqref{eq:copycompare} states that copying public information and then immediately comparing yields the identity.

The following theorem demonstrates that these structures are easy to work with in \cat{2Rel}.
\begin{theorem}
Every 0\-cell in \cat{2Rel} carries structures \eqref{eq:copypublicdata}--\eqref{eq:createpublicdata} satisfying equations~\eqref{eq:topological1}--\eqref{eq:copycompare} in an essentially unique way.
\end{theorem}
\begin{proof}[Proof sketch]
A 1\-cell $A:1 \to S$ is determined by an $S$-indexed family of finite sets $A_s$, and its isomorphism class is determined by the cardinalities of those sets. Every such 1\-cell has an ambidextrous adjoint, meaning precisely that values can be given for structures \eqref{eq:copypublicdata}--\eqref{eq:createpublicdata} that satisfy equations \eqref{eq:topological1}--\eqref{eq:topological2}. The result is a Frobenius algebra structure~\cite{l06-faaa}, which will be commutative exactly when each of the finite sets $A_s$ has cardinality 1, which satisfies the equations labelled~\eqref{eq:commplanar}. The resulting structures automatically satisfy equation~\eqref{eq:copycompare}.
\end{proof}

\noindent
Indeed, such a structure in \cat{2Rel} gives rise to a commutative dagger-Frobenius algebra in \cat{Rel}, corresponding to a discrete groupoid with respect to the classification of such structures as abelian groupoids~\cite{p09-qcs, hcc12-rfa}. This suggests an expansion of our formalism to the case where objects of the bicategory are arbitrary abelian groupoids. It would be interesting to consider what procedures in classical information might be naturally modelled by such an extended formalism.

\subsection{Interacting private and public data}

\noindent
Interesting phenomena arise when we study interactions between public and private information. There are three basic forms that such an interaction can take: converting private data to public data; converting public data to private data; and using public data to modify private data.

Conversion processes between public and private data take the following forms:
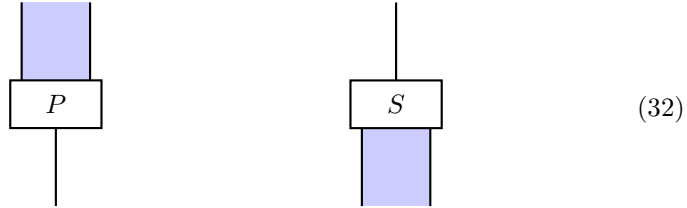
\begin{calign}
\begin{aligned}
\begin{tikzpicture}[scale=\diagramscale, thick]
    \begin{pgfonlayer}{foreground}
    \node (M) [minimum width=\componentwidth, minimum height=18pt, fill=white, draw=black, thick, anchor=center] at (1.0,1.0) {$P$};
    \end{pgfonlayer}
    \draw [fill=\fillA] (0.5,2.5)
        to (0.5,1.0)
        to (1.5,1.0)
        to (1.5,2.5);
    \draw (1,-0.5) to (1,1.0);
\end{tikzpicture}
\end{aligned}
&
\begin{aligned}
\begin{tikzpicture}[scale=\diagramscale, thick, yscale=-1]
    \begin{pgfonlayer}{foreground}
    \node (M) [minimum width=\componentwidth, minimum height=18pt, fill=white, draw=black, thick, anchor=center] at (1.0,1.0) {$S$};
    \end{pgfonlayer}
    \draw [fill=\fillA] (0.5,2.5)
        to (0.5,1.0)
        to (1.5,1.0)
        to (1.5,2.5);
    \draw (1,-0.5) to (1,1.0);
\end{tikzpicture}
\end{aligned}
\end{calign}
Here $P$ is a publication process converting private data into public data, and $S$ is a sampling process converting public data into private data. Their interpretations rests entirely on their types; there are no equations which we require them to satisfy. These processes need not be deterministic, or invertible, in general. We could also allow them to have a kernel, meaning that the computation will halt on some inputs.

The final type of process we introduce is the controlled computation, which performs an operation on private data depending on the value of some public data:
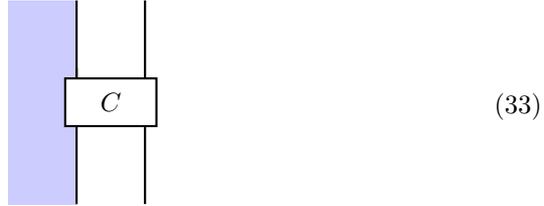
\begin{equation}
\begin{aligned}
\begin{tikzpicture}[scale=\diagramscale, thick]
    \begin{pgfonlayer}{foreground}
    \node (M) [minimum width=\componentwidth, minimum height=18pt, fill=white, draw=black, thick, anchor=center] at (1.0,1.0) {$C$};
    \end{pgfonlayer}
    \draw [fill=\fillA, draw=none] (0.5,2.5) rectangle (-0.5,-0.5);
    \draw (0.5,2.5)
        to (0.5,1.0)
        to (1.5,1.0)
        to (1.5,2.5);
    \draw (0.5,-0.5) to (0.5,1.5);
    \draw (1.5,1.6 |- M.south)
        to (1.5,-0.5);
\end{tikzpicture}
\end{aligned}
\end{equation}
Such an operation can modify the private data, but not the public data.
\begin{lemma}
A controlled computation cannot modify public data.
\end{lemma}
\begin{proof}
We can use the topological behaviour of public information to rewrite our controlled computation vertex $C$ in the following way:
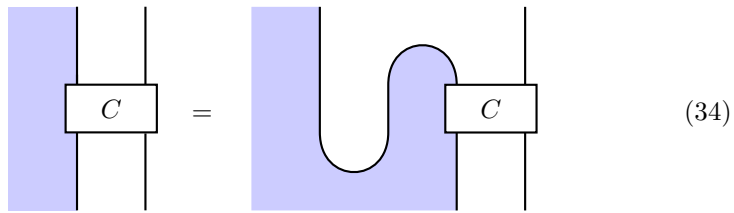
\begin{equation}
\begin{aligned}
\begin{tikzpicture}[scale=\diagramscale, thick]
    \begin{pgfonlayer}{foreground}
    \node (M) [minimum width=\componentwidth, minimum height=18pt, fill=white, draw=black, thick, anchor=center] at (1.0,1.0) {$C$};
    \end{pgfonlayer}
    \draw [fill=\fillA, draw=none] (0.5,2.5) rectangle (-0.5,-0.5);
    \draw (0.5,2.5)
        to (0.5,1.0)
        to (1.5,1.0)
        to (1.5,2.5);
    \draw (0.5,-0.5) to (0.5,1.5);
    \draw (1.5,1.6 |- M.south)
        to (1.5,-0.5);
\end{tikzpicture}
\end{aligned}
\quad=\quad
\begin{aligned}
\begin{tikzpicture}[scale=\diagramscale, thick]
    \draw [fill=\fillA, draw=none] (1.5,-2.75)
        to (1.5,-0.9)
        to [out=up, in=up, looseness=2] (0.5,-0.9)
        to (0.5,-1.6)
        to [out=down, in=down, looseness=2] (-0.5,-1.6)
        to (-0.5,0.25)
        to (-1.5,0.25)
        to (-1.5,-2.75);
    \draw [] (1.5,-2.75)
        to (1.5,-0.9)
        to [out=up, in=up, looseness=2] (0.5,-0.9)
        to (0.5,-1.6)
        to [out=down, in=down, looseness=2] (-0.5,-1.6)
        to (-0.5,0.25);
\draw (2.5,-2.75) to (2.5,0.25);
    \node (M) [minimum width=\componentwidth, minimum height=18pt, fill=white, draw=black, thick, anchor=center] at (2,-1.25) {$C$};
\end{tikzpicture}
\end{aligned}
\end{equation}
In this form it is clear that the public data is not modified, since it is explicitly copied before $C$ is implemented.
\end{proof}

\noindent
This result fits well with our intuition about public data as a being carried by a large, correlated family of systems. To change the value of the public data would require modifying all of these systems, but the process $C$ only has access to a restricted subset, as made explicit by the open boundary on the left-hand side of the diagram.

\section{Modelling cryptographic procedures}
\label{sec:classicalprocedures}

\subsection{Encrypted communication}

Suppose Alice is sending an encrypted message to Bob.  We use a 2-cell $E$ to represent Alice's encryption process, which relates the private plaintext $P$ and the private key $K$ to the public ciphertext $C$:
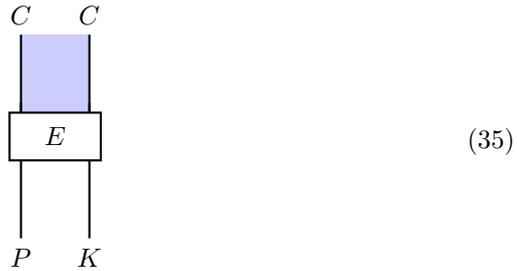
\begin{equation}
\begin{aligned}
\begin{tikzpicture}[scale=\diagramscale, thick]
      \begin{pgfonlayer}{foreground}
      \node (M) [minimum width=\componentwidth, minimum height=18pt, fill=white, draw=black, thick, anchor=center] at (1.0,1.0) {$E$};
      \end{pgfonlayer}
      \draw [fill=\fillA] (0.5,2.5)
            node [above] {$C$}
          to (0.5,1.0)
          to (1.5,1.0)
          to (1.5,2.5)
            node [above] {$C$};
      \node (P) [below] at (0.5,-0.5) {$P$};
      \draw (P) to (0.5,1.5);
      \node (K) [below] at (1.5,-0.5) {$K$};
      \draw (K) to (1.5,1.5);
  \end{tikzpicture}
  \end{aligned}
\end{equation}

\noindent
Similarly, we represent Bob's decryption process $D$ as a 2-cell that relates the public ciphertext and private key to the same ciphertext and a private plaintext.
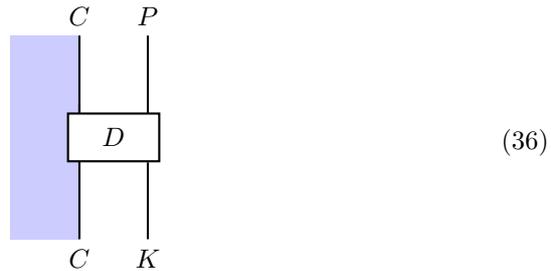
\begin{equation}
\begin{aligned}
  \begin{tikzpicture}[scale=\diagramscale, thick]
      \begin{pgfonlayer}{foreground}
      \node (M) [minimum width=\componentwidth, minimum height=18pt, fill=white, draw=black, thick, anchor=center] at (1.0,1.0) {$D$};
      \end{pgfonlayer}
      \draw [fill=\fillA, draw=none] (0.5,2.5) rectangle (-0.5,-0.5);
      \begin{pgfonlayer}{foreground}
      \end{pgfonlayer}
      \node (P) [above] at (1.5,2.5) {$P$};
      \draw (0.5,2.5)
          to (0.5,1.0)
          to (1.5,1.0)
          to (P);
      \draw (0.5,-0.5) to (0.5,1.5);
      \node (K) [below] at (1.5,-0.5) {$K$};
      \draw (1.5,1.6 |- M.south)
          to (K);
      \node [above] at (0.5,2.5) {$C$};
      \node [below] at (0.5,-0.5) {$C$};
  \end{tikzpicture}
  \end{aligned}
\end{equation}

\noindent
Encryption and decryption are deterministic; key generation is not.  We represent key generation as a special 2-cell, the curried identity relation on the set of keys $K$.
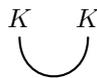
\begin{center}
  \begin{tikzpicture}[scale=\diagramscale, thick]
    \draw (0,0.5) node [above] {$K$} to [out=down,in=down, looseness=2] (1,0.5) node [above] {$K$};
  \end{tikzpicture}
\end{center}
This is the unit morphism for a self-duality on $K$, as described in Section~\ref{sec:privateinformation}.

Using our topological language, we can express correctness of encrypted communication in the following way:
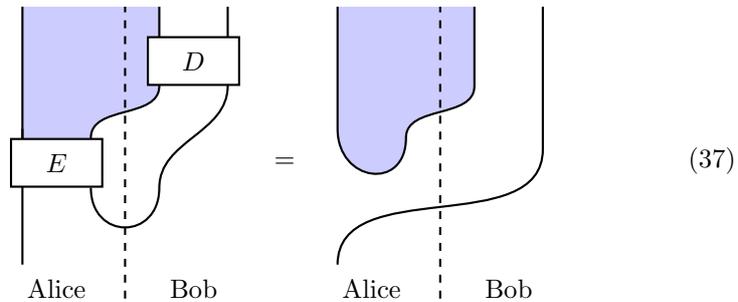
\begin{equation}
\label{eq:encryptedcommunication}
\begin{aligned}
\begin{tikzpicture}[scale=\diagramscale, thick]
    \begin{pgfonlayer}{foreground}
    \node (U) [minimum width=\componentwidth, minimum height=18pt, fill=white, draw=black, thick] at (3,2.5) {$D$};
    \node (M) [minimum width=\componentwidth, minimum height=18pt, fill=white, draw=black, thick, anchor=center] at (1.0,1.0) {$E$};
    \end{pgfonlayer}
    \draw [fill=\fillA] (0.5,3.3)
        to (0.5,1.0)
        to (1.5,1.0) to (1.5,2 |- M.north)
        to [out=up, in=down] (2.5,2.0 |- U.south) to (2.5,3.3);
    \draw (0.5,-0.5) to (0.5,1.5);
    \draw (1.5,1.6 |- M.south)
        to [out=down, in=down, looseness=2] (2.5,0.5 |- M.south) to [out=up, in=down] (3.5,2 |- U.south)
        to (3.5,3.3);
    \draw [dashed] (2.0,-1) to (2.0,3.3);
    \node at (1,-0.85) {Alice};
    \node at (3,-0.85) {Bob};
\end{tikzpicture}
\end{aligned}
\quad=\quad
\begin{aligned}
\begin{tikzpicture}[scale=\diagramscale, thick]
    \begin{pgfonlayer}{foreground}
    \end{pgfonlayer}
    \draw [fill=\fillA] (0.5,3.3)
        to (0.5,1.5)
        to [out=down, in=down, looseness=2] (1.5,2 |- M.north)
        to [out=up, in=down] (2.5,2.0 |- U.south) to (2.5,3.3);
    \draw (0.5,-0.5) to [out=up, in=down] (3.5,1.2) to (3.5,3.3);
    \draw [dashed] (2.0,-1) to (2.0,3.3);
    \node at (1,-0.85) {Alice};
    \node at (3,-0.85) {Bob};
\end{tikzpicture}
\end{aligned}
\end{equation}
This is the same 2\-dimensional equation as that used in \cite{v12-hsqp} to describe quantum teleportation. The encryption step takes the place of the measurement operation, and the decryption step takes the place of the controlled unitary correction.  The ciphertext takes the place of the classical bits transmitted from Alice to Bob.  This provides an intuition for why no faster-than-light communication is possible with entangled particles: Alice and Bob merely share a quantum variant of a one-time pad, and the actual encoded message must still be sent at some finite speed.

The simplest nontrivial implementation of this protocol is the encrypted communication of a single bit. We can describe concretely the values of $E$, $D$ and the key creation step~$\eta$  as 2\-cells in \cat{2Rel} which correspond to this scenario. We choose $C=P=K$ to be the 2\-element set, and the 2\-cells take the following values:
\begin{align}
E &= \Bigg(\begin{pmatrix}1 & 0 & 0 & 1 \\ 0 & 1 & 1 & 0\end{pmatrix} \Bigg)
\\
D &= \begin{pmatrix}
\begin{pmatrix}
1 & 0 \\ 0 & 1
\end{pmatrix}
\\
\begin{pmatrix}
0 & 1 \\ 1 & 0
\end{pmatrix}
\end{pmatrix}
\\
\eta &= \Big( \begin{pmatrix} 1 & 0 & 0 & 1 \end{pmatrix} \Big)
\end{align}
Here $E$ is a matrix containing a single relation from a 4\-element set to a 2\-element set, which is exactly the multiplication operation for the group $\mathbb{Z}_2$; $D$ is matrix of invertible single-bit operations to apply depending on which bit is published at the encryption step; and $\eta$ is a matrix with a single entry, the relation representing nondeterministic creation of the pair of keys $(0,0)$ or $(1,1)$. Using the definition of the bicategory \cat{2Rel}, it can be checked that these values satisfy equation~\eqref{eq:encryptedcommunication}.

However, our formalism allows us to carry out an analysis of the protocol in its abstract form, and hence draw conclusions which will apply to any particular implementation. To focus on its algebraic properties, we can simplify equation~\eqref{eq:encryptedcommunication} topologically in the following way:
\begin{equation}
\label{eq:simplifiedencryptedcommunication}
\begin{aligned}
\begin{tikzpicture}[scale=\diagramscale, thick]
    \begin{pgfonlayer}{foreground}
    \node (U) [minimum width=\componentwidth, minimum height=18pt, fill=white, draw=black, thick] at (2,2.0) {$D$};
    \node (M) [minimum width=\componentwidth, minimum height=18pt, fill=white, draw=black, thick, anchor=center] at (1.0,1.0) {$E$};
    \end{pgfonlayer}
    \draw [fill=\fillA] (0.5,3)
        to (0.5,1.0)
        to (1.5,1.0)
        to (1.5,3);
    \draw (0.5,-0.75) to (0.5,1.5);
    \draw (1.5,1.6 |- M.south)
        to [out=down, in=up, looseness=1.5] (1.5,0.5)
        to [out=down, in=down, looseness=2] (2.5,0.5)
        to (2.5,3.0);
\end{tikzpicture}
\end{aligned}
\quad=\quad
\begin{aligned}
\begin{tikzpicture}[scale=\diagramscale, thick]
    \begin{pgfonlayer}{foreground}
    \end{pgfonlayer}
    \draw [fill=\fillA] (0.5,3.0)
        to (0.5,1.5)
        to [out=down, in=down, looseness=2] (1.5,1.5)
        to (1.5,3.0);
    \draw (2.5,-0.75) to (2.5,3.0);
\end{tikzpicture}
\end{aligned}
\end{equation}
We can describe a variety of security properties in a graphical way. Here is the first, which is the primary security property for encrypted communication:
\begin{equation}
\label{eq:security1}
\begin{aligned}
\begin{tikzpicture}[scale=\diagramscale, thick]
    \begin{pgfonlayer}{foreground}
    \node (M) [minimum width=\componentwidth, minimum height=18pt, fill=white, draw=black, thick, anchor=center] at (1.0,1.0) {$E$};
    \end{pgfonlayer}
    \draw [fill=\fillA] (0.5,2)
        to (0.5,1.0)
        to (1.5,1.0)
        to (1.5,2);
    \draw (0.5,-0.5) to (0.5,1.5);
    \draw (1.5,1.6 |- M.south)
        to [out=down, in=up, looseness=1.5] (1.5,0.5)
        to [out=down, in=down, looseness=2] (2.5,0.5)
        to (2.5,1.0) node [Vertex] {};
\end{tikzpicture}
\end{aligned}
\quad=\quad
\begin{aligned}
\begin{tikzpicture}[scale=\diagramscale, thick]
    \begin{pgfonlayer}{foreground}
    \end{pgfonlayer}
    \draw [fill=\fillA] (0.5,2)
        to (0.5,1.5)
        to [out=down, in=down, looseness=2] (1.5,1.5)
        to (1.5,2);
    \draw (1,-0.5) to (1,0.2) node [Vertex] {};
\end{tikzpicture}
\end{aligned}
\end{equation}
This says that if we encrypt a message using one copy of a one-time pad, and then delete the other copy of the one-time pad, this is equivalent to deleting our original message and producing a random ciphertext. So in particular, deleting the key causes the original message to be unrecoverable. This also ensures that the whole space of possible keys is being used.

We can use our formalism to derive from this security property a strong constraint on the encryption operation $E$.
\begin{theorem}
If the encryption step in classical encrypted communication satisfies property~\eqref{eq:security1}, then encryption is not invertible unless the space of messages is trivial.
\end{theorem}
\begin{proof}
Suppose encryption is invertible. Then composing both sides of~\eqref{eq:security1} with $E ^{-1}$ gives the following graphical expression:
\begin{equation}
\begin{aligned}
\begin{tikzpicture}[scale=\diagramscale, thick]
    \draw (0.5,-0.5) to (0.5,2.7);
    \draw (1.5,2.7)
        to [out=down, in=up, looseness=1.5] (1.5,0.5)
        to [out=down, in=down, looseness=2] (2.5,0.5)
        to (2.5,1.0) node [Vertex] {};
\end{tikzpicture}
\end{aligned}
\quad=\quad
\begin{aligned}
\begin{tikzpicture}[scale=\diagramscale, thick]
    \begin{pgfonlayer}{foreground}
    \node (M) [minimum width=\componentwidth, minimum height=18pt, fill=white, draw=black, thick, anchor=south] at (1.0,1.5) {$E ^{-1}$};
    \end{pgfonlayer}
    \draw [fill=\fillA] (0.5,2)
        to (0.5,1.5)
        to [out=down, in=down, looseness=2] (1.5,1.5)
        to (1.5,2);
    \draw (1,-0.5) to (1,0.2) node [Vertex] {};
    \draw (0.5,2) to (0.5,2.7);
    \draw (1.5,2) to (1.5,2.7);
\end{tikzpicture}
\end{aligned}
\end{equation}
Hence the identity process on the set of messages factors through the one-element set.
\end{proof}

We can draw a quite different conclusion for the decryption process $D$.
\begin{theorem}
In classical encrypted communication, the decryption step is invertible.
\end{theorem}
\begin{proof}
From equation~\eqref{eq:simplifiedencryptedcommunication} representing correctness of encrypted communication, we apply the topological properties of public information to obtain the following equivalent equation:
\begin{equation}
\label{eq:encryptedcommunication2}
\begin{aligned}
\begin{tikzpicture}[scale=\diagramscale, thick]
    \begin{pgfonlayer}{foreground}
    \node (U) [minimum width=\componentwidth, minimum height=18pt, fill=white, draw=black, thick] at (2,2.0) {$D$};
    \node (M) [minimum width=\componentwidth, minimum height=18pt, fill=white, draw=black, thick, anchor=center] at (1.0,1.0) {$E$};
    \end{pgfonlayer}
    \draw [fill=\fillA, draw=none] (1.5,3)
        to (1.5,1.0)
        to (0.5,1)
        to (0.5,1.5)
        to [out=up, in=up, looseness=2] (-0.5,1.5)
        to (-0.5,-0.75)
        to (-1.5,-0.75)
        to (-1.5,3);
    \draw [] (1.5,3)
        to (1.5,1.0)
        to (0.5,1)
        to (0.5,1.5)
        to [out=up, in=up, looseness=2] (-0.5,1.5)
        to (-0.5,-0.75);
    \draw (0.5,-0.75) to (0.5,1.5);
    \draw (1.5,1.6 |- M.south)
        to [out=down, in=up, looseness=1.5] (1.5,0.5)
        to [out=down, in=down, looseness=2] (2.5,0.5)
        to (2.5,3.0);
\end{tikzpicture}
\end{aligned}
\quad=\quad
\begin{aligned}
\begin{tikzpicture}[scale=\diagramscale, thick]
    \begin{pgfonlayer}{foreground}
    \end{pgfonlayer}
    \draw [fill=\fillA, draw=none] (0.5,3.0)
        to (0.5,-0.75)
        to (1.5,-0.75)
        to (1.5,3.0);
    \draw (2.5,-0.75) to (2.5,3.0);
    \draw (1.5,3) to (1.5,-0.75);
\end{tikzpicture}
\end{aligned}
\end{equation}
This says that $D$ has a right inverse given by $E$ with its top-left and bottom-right legs twisted in the manner indicated. However, by Theorem~\ref{lem:endoinvertible}, if an endomorphism is a left inverse then it must also be a right inverse, and hence our theorem follows, with the following expression for $D ^{-1}$:
\begin{equation}
\label{eq:dinverse}
\begin{aligned}
\begin{tikzpicture}[scale=\diagramscale, thick]
    \begin{pgfonlayer}{foreground}
    \end{pgfonlayer}
    \draw [fill=\fillA, draw=none] (0.5,2.75)
        to (0.5,-0.75)
        to (1.5,-0.75)
        to (1.5,2.75);
    \draw (2.5,-0.75) to (2.5,2.75);
    \draw (1.5,2.75) to (1.5,-0.75);
    \node (U) [minimum width=\componentwidth, minimum height=18pt, fill=white, draw=black, thick] at (2,1.0) {$D ^{-1}$};
\end{tikzpicture}
\end{aligned}
\quad=\quad
\begin{aligned}
\begin{tikzpicture}[scale=\diagramscale, thick]
    \begin{pgfonlayer}{foreground}
    \node (M) [minimum width=\componentwidth, minimum height=18pt, fill=white, draw=black, thick, anchor=center] at (1.0,1.0) {$E$};
    \end{pgfonlayer}
    \draw [fill=\fillA, draw=none] (1.5,2.75)
        to (1.5,1.0)
        to (0.5,1)
        to (0.5,1.5)
        to [out=up, in=up, looseness=2] (-0.5,1.5)
        to (-0.5,-0.75)
        to (-1.5,-0.75)
        to (-1.5,2.75);
    \draw [] (1.5,2.75)
        to (1.5,1.0)
        to (0.5,1)
        to (0.5,1.5)
        to [out=up, in=up, looseness=2] (-0.5,1.5)
        to (-0.5,-0.75);
    \draw (0.5,-0.75) to (0.5,1.5);
    \draw (1.5,1.6 |- M.south)
        to [out=down, in=up, looseness=1.5] (1.5,0.5)
        to [out=down, in=down, looseness=2] (2.5,0.5)
        to (2.5,2.75);
\end{tikzpicture}
\end{aligned}
\end{equation}
\end{proof}

\noindent
It follows that we can reconstruct $E$ from the knowledge of $D$ and its inverse.
\begin{theorem}
For an implementation of classical encrypted communication, we have
\begin{equation}
\label{eq:EintermsofDinv}
\begin{aligned}
\begin{tikzpicture}[scale=\diagramscale, thick]
    \begin{pgfonlayer}{foreground}
    \end{pgfonlayer}
    \draw (2.5,-0.75) to (2.5,2.75);
    \draw (1.5,2.75) to (1.5,-0.75);
    \draw [fill=\fillA] (1.5,2.75) to (1.5,1) to (2.5,1) to (2.5,2.75);
    \node (U) [minimum width=\componentwidth, minimum height=18pt, fill=white, draw=black, thick] at (2,1.0) {$E$};
\end{tikzpicture}
\end{aligned}
\quad=\quad
\begin{aligned}
\begin{tikzpicture}[scale=\diagramscale, thick, yscale=-1]
    \begin{pgfonlayer}{foreground}
    \node (M) [minimum width=\componentwidth, minimum height=18pt, fill=white, draw=black, thick, anchor=center] at (1.0,1.0) {$D ^{-1}$};
    \end{pgfonlayer}
    \draw [] (1.5,2.75)
        to (1.5,1.0);
    \draw [fill=\fillA] (0.5,-0.75)
        to (0.5,1.5)
        to [out=up, in=up, looseness=2] (-0.5,1.5)
        to (-0.5,-0.75);
    \draw (1.5,1.6 |- M.south)
        to [out=down, in=up, looseness=1.5] (1.5,0.5)
        to [out=down, in=down, looseness=2] (2.5,0.5)
        to (2.5,2.75);
\end{tikzpicture}
\end{aligned}
\end{equation}
\end{theorem}
\begin{proof}
We apply the topological properties of public information to expression~\eqref{eq:dinverse} to obtain the following:
{\def\diagramscale{0.86}
\begin{calign}
\begin{aligned}
\begin{tikzpicture}[scale=\diagramscale, thick, yscale=-1]
    \begin{pgfonlayer}{foreground}
    \node (M) [minimum width=\componentwidth, minimum height=18pt, fill=white, draw=black, thick, anchor=center] at (1.0,1.0) {$D ^{-1}$};
    \end{pgfonlayer}
    \draw [] (1.5,2.75)
        to (1.5,1.0);
    \draw [fill=\fillA] (0.5,-0.75)
        to (0.5,1.5)
        to [out=up, in=up, looseness=2] (-0.5,1.5)
        to (-0.5,-0.75);
    \draw (1.5,1.6 |- M.south)
        to [out=down, in=up, looseness=1.5] (1.5,0.5)
        to [out=down, in=down, looseness=2] (2.5,0.5)
        to (2.5,2.75);
\end{tikzpicture}
\end{aligned}
\hspace{5pt}=\hspace{5pt}
\begin{aligned}
\begin{tikzpicture}[scale=\diagramscale, thick]
    \begin{pgfonlayer}{foreground}
    \node (M) [minimum width=\componentwidth, minimum height=18pt, fill=white, draw=black, thick, anchor=center] at (1.0,1.0) {$E$};
    \end{pgfonlayer}
    \draw [fill=\fillA, draw=none] (1.5,2.75)
        to (1.5,1.0)
        to (0.5,1)
        to (0.5,1.5)
        to [out=up, in=up, looseness=2] (-0.5,1.5)
        to (-0.5,0.5)
        to [out=down, in=down, looseness=2] (-1.5,0.5)
        to (-1.5,2.75);
    \draw [] (1.5,2.75)
        to (1.5,1.0)
        to (0.5,1)
        to (0.5,1.5)
        to [out=up, in=up, looseness=2] (-0.5,1.5)
        to (-0.5,0.5) to [out=down, in=down, looseness=2] (-1.5,0.5) to (-1.5,2.75);
    \draw (0.5,-0.75) to (0.5,1.5);
    \draw (1.5,1.6 |- M.south)
        to [out=down, in=up, looseness=1.5] (1.5,0.5)
        to [out=down, in=down, looseness=2] (2.5,0.5)
        to (2.5,1.5) to [out=up, in=up, looseness=2] (3.5,1.5) to (3.5,-0.75);
\end{tikzpicture}
\end{aligned}
\end{calign}}The right-hand side of this expression evaluates to $E$, by the topological properties~\eqref{eq:topological1} of 2\-dimensional regions and the snake equations~\eqref{eq:snakeequations}.
\end{proof}

While property~\eqref{eq:security1} is primary, there are other security properties of the encryption process that we could consider. The first states that if we encode with a random key, this is equivalent to deleting the original message and producing random ciphertext:
\begin{equation}
\label{eq:security2}
\begin{aligned}
\begin{tikzpicture}[scale=\diagramscale, thick]
    \begin{pgfonlayer}{foreground}
    \node (M) [minimum width=\componentwidth, minimum height=18pt, fill=white, draw=black, thick, anchor=center] at (1.0,1.0) {$E$};
    \end{pgfonlayer}
    \draw [fill=\fillA] (0.5,2)
        to (0.5,1.0)
        to (1.5,1.0)
        to (1.5,2);
    \draw (0.5,-0.5) to (0.5,1.5);
    \draw (1.5,1.6 |- M.south)
        to (1.5,0.2) node [Vertex] {};
\end{tikzpicture}
\end{aligned}
\quad=\quad
\begin{aligned}
\begin{tikzpicture}[scale=\diagramscale, thick]
    \begin{pgfonlayer}{foreground}
    \end{pgfonlayer}
    \draw [fill=\fillA] (0.5,2)
        to (0.5,1.5)
        to [out=down, in=down, looseness=2] (1.5,1.5)
        to (1.5,2);
    \draw (1,-0.5) to (1,0.2) node [Vertex] {};
\end{tikzpicture}
\end{aligned}
\end{equation}
Secondly, we could encode a random message with a specified key:
\begin{equation}
\label{eq:security3}
\begin{aligned}
\begin{tikzpicture}[scale=\diagramscale, thick]
    \begin{pgfonlayer}{foreground}
    \node (M) [minimum width=\componentwidth, minimum height=18pt, fill=white, draw=black, thick, anchor=center] at (1.0,1.0) {$E$};
    \end{pgfonlayer}
    \draw [fill=\fillA] (0.5,2)
        to (0.5,1.0)
        to (1.5,1.0)
        to (1.5,2);
    \draw (1.5,-0.5) to (1.5,1.5);
    \draw (0.5,1.6 |- M.south)
        to (0.5,0.2) node [Vertex] {};
\end{tikzpicture}
\end{aligned}
\quad=\quad
\begin{aligned}
\begin{tikzpicture}[scale=\diagramscale, thick]
    \begin{pgfonlayer}{foreground}
    \end{pgfonlayer}
    \draw [fill=\fillA] (0.5,2)
        to (0.5,1.5)
        to [out=down, in=down, looseness=2] (1.5,1.5)
        to (1.5,2);
    \draw (1,-0.5) to (1,0.2) node [Vertex] {};
\end{tikzpicture}
\end{aligned}
\end{equation}
This property says that this is the same as deleting the key, and producing a random ciphertext.

We can also consider security properties for the decryption process. 
\begin{equation}
\label{eq:security4}
\begin{aligned}
\begin{tikzpicture}[scale=\diagramscale, thick]
    \begin{pgfonlayer}{foreground}
    \node (M) [minimum width=\componentwidth, minimum height=18pt, fill=white, draw=black, thick, anchor=center] at (1.0,1.0) {$D$};
    \end{pgfonlayer}
    \draw [fill=\fillA, draw=none] (0.5,2) rectangle (-0.5,-0.5);
    \draw (0.5,2)
        to (0.5,1.0)
        to (1.5,1.0)
        to (1.5,2);
    \draw (0.5,-0.5) to (0.5,1.5);
    \draw (1.5,1.6 |- M.south)
        to (1.5,0.25) node [Vertex] {};
\end{tikzpicture}
\end{aligned}
\quad=\quad
\begin{aligned}
\begin{tikzpicture}[scale=\diagramscale, thick]
    \draw [fill=\fillA, draw=none] (0.5,2) rectangle (-0.5,-0.5);
    \draw (0.5,-0.5) to (0.5,2.0);
    \draw (1.5,2)
        to (1.5,0.75) node [Vertex] {};
\end{tikzpicture}
\end{aligned}
\end{equation}
This says that if an attacker chooses nondeterministically from the space of all possible keys, every possible message can be produced, regardless of the ciphertext. So if an attacker has no knowledge of the key, they cannot extract information from the ciphertext.

In fact, we can use our formalism to show that all of these security properties follow from the primary security property~\eqref{eq:security1}.
\begin{theorem}
In classical encrypted communication, \eqref{eq:security1} implies \eqref{eq:security2}, \eqref{eq:security3} and \eqref{eq:security4}.
\end{theorem}
\begin{proof}
The implication $\eqref{eq:security1} \Rightarrow \eqref{eq:security2}$ follows from the topological property~\eqref{eq:twistdelete1} of the deletion map. For the other implications, we compose expression~\eqref{eq:dinverse} for $D ^{-1}$ with the deletion map at the top-right leg, obtaining the following:
\begin{equation}
\nonumber
\begin{aligned}
\begin{tikzpicture}[scale=0.8, thick]
    \begin{pgfonlayer}{foreground}
    \end{pgfonlayer}
    \draw [fill=\fillA, draw=none] (0.5,2.75)
        to (0.5,-0.75)
        to (1.5,-0.75)
        to (1.5,2.75);
    \draw (2.5,-0.75) to (2.5,2.0) node [Vertex] {};
    \draw (1.5,2.75) to (1.5,-0.75);
    \node (U) [minimum width=\componentwidth, minimum height=18pt, fill=white, draw=black, thick] at (2,1.0) {$D ^{-1}$};
\end{tikzpicture}
\end{aligned}
\quad=\quad
\begin{aligned}
\begin{tikzpicture}[scale=0.8, thick]
    \begin{pgfonlayer}{foreground}
    \node (M) [minimum width=\componentwidth, minimum height=18pt, fill=white, draw=black, thick, anchor=center] at (1.0,1.0) {$E$};
    \end{pgfonlayer}
    \draw [fill=\fillA, draw=none] (1.5,2.75)
        to (1.5,1.0)
        to (0.5,1)
        to (0.5,1.5)
        to [out=up, in=up, looseness=2] (-0.5,1.5)
        to (-0.5,-0.75)
        to (-1.5,-0.75)
        to (-1.5,2.75);
    \draw [] (1.5,2.75)
        to (1.5,1.0)
        to (0.5,1)
        to (0.5,1.5)
        to [out=up, in=up, looseness=2] (-0.5,1.5)
        to (-0.5,-0.75);
    \draw (0.5,-0.75) to (0.5,1.5);
    \draw (1.5,1.6 |- M.south)
        to (1.5,0.5)
        to [out=down, in=down, looseness=2] (2.5,0.5)
        to (2.5,2.0) node [Vertex] {};
\end{tikzpicture}
\end{aligned}
\end{equation}
\begin{equation}
\label{eq:2implies5}
\hspace{30pt}\quad=\quad
\begin{aligned}
\begin{tikzpicture}[scale=0.8, thick]
    \draw [fill=\fillA, draw=none] (1.5,2.75)
        to (1.5,1.0)
        to [out=down, in=down, looseness=2] (0.5,1)
        to (0.5,1.5)
        to [out=up, in=up, looseness=2] (-0.5,1.5)
        to (-0.5,-0.75)
        to (-1.5,-0.75)
        to (-1.5,2.75);
    \draw [] (1.5,2.75)
        to (1.5,1.0)
        to [out=down, in=down, looseness=2] (0.5,1)
        to (0.5,1.5)
        to [out=up, in=up, looseness=2] (-0.5,1.5)
        to (-0.5,-0.75);
    \draw (1.0,0) node [Vertex] {}
        to (1.0,-0.75);
\end{tikzpicture}
\end{aligned}
\quad=\quad
\begin{aligned}
\begin{tikzpicture}[scale=0.8, thick]
    \draw [fill=\fillA, draw=none] (-0.5,2.75)
        to (-0.5,-0.75)
        to (-1.5,-0.75)
        to (-1.5,2.75);
    \draw [] (-0.5,2.75)
        to (-0.5,-0.75);
    \draw (0.3,1) node [Vertex] {}
        to (0.3,-0.75);
\end{tikzpicture}
\end{aligned}
\end{equation}
Every invertible 2\-cell in \cat{Rel} is a family of bijections, and hence its converse is its inverse. Taking the converse is a functorial operation, and so taking the converse of of the first and last diagram here, we obtain property~\eqref{eq:security4}:
\begin{equation}
\tag{\ref{eq:security4}}
\begin{aligned}
\begin{tikzpicture}[scale=0.9, yscale=1, thick]
    \begin{pgfonlayer}{foreground}
    \node (M) [minimum width=\componentwidth, minimum height=18pt, fill=white, draw=black, thick, anchor=center] at (1.0,1.0) {$D$};
    \end{pgfonlayer}
    \draw [fill=\fillA, draw=none] (0.5,2) rectangle (-0.5,-0.5);
    \draw (0.5,2)
        to (0.5,1.0)
        to (1.5,1.0)
        to (1.5,2);
    \draw (0.5,-0.5) to (0.5,1.5);
    \draw (1.5,1.6 |- M.south)
        to (1.5,0.15) node [Vertex] {};
\end{tikzpicture}
\end{aligned}
\quad=\quad
\begin{aligned}
\begin{tikzpicture}[scale=0.9, thick]
    \draw [fill=\fillA, draw=none] (0.5,2) rectangle (-0.5,-0.5);
    \draw (0.5,-0.5) to (0.5,2.0);
    \draw (1.5,2)
        to (1.5,0.75) node [Vertex] {};
\end{tikzpicture}
\end{aligned}
\end{equation}
For the final property~\eqref{eq:security3}, we postcompose this expression with the 2\-cell $D ^{-1}$, obtaining the following expression:
\begin{equation}
\begin{aligned}
\begin{tikzpicture}[scale=0.9, yscale=1, thick]
    \draw [fill=\fillA, draw=none] (0.5,2) rectangle (-0.5,-0.5);
    \draw (0.5,-0.5) to (0.5,2.0);
    \draw (1.5,2)
        to (1.5,0.75) node [Vertex] {};
\end{tikzpicture}
\end{aligned}
\quad=\quad
\begin{aligned}
\begin{tikzpicture}[scale=0.9, yscale=1, thick]
    \begin{pgfonlayer}{foreground}
    \node (M) [minimum width=\componentwidth, minimum height=18pt, fill=white, draw=black, thick, anchor=center] at (1.0,1.0) {$D^{-1}$};
    \end{pgfonlayer}
    \draw [fill=\fillA, draw=none] (0.5,2) rectangle (-0.5,-0.5);
    \draw (0.5,2)
        to (0.5,1.0)
        to (1.5,1.0)
        to (1.5,2);
    \draw (0.5,-0.5) to (0.5,1.5);
    \draw (1.5,1.6 |- M.south)
        to (1.5,0.15) node [Vertex] {};
\end{tikzpicture}
\end{aligned}
\end{equation}
We can use this to prove security property~\eqref{eq:security3}, where we also make use of expression~\eqref{eq:EintermsofDinv} giving $E$ in terms of $D ^{-1}$:
\begin{equation}
\nonumber
\begin{aligned}
\begin{tikzpicture}[scale=\diagramscale, thick]
    \begin{pgfonlayer}{foreground}
    \end{pgfonlayer}
    \draw (2.5,-0.75) to (2.5,2.75);
    \draw (1.5,2.75) to (1.5,0) node [Vertex] {};
    \draw [fill=\fillA] (1.5,2.75) to (1.5,1) to (2.5,1) to (2.5,2.75);
    \node (U) [minimum width=\componentwidth, minimum height=18pt, fill=white, draw=black, thick] at (2,1.0) {$E$};
\end{tikzpicture}
\end{aligned}
\quad=\quad
\begin{aligned}
\begin{tikzpicture}[scale=\diagramscale, thick, yscale=-1]
    \begin{pgfonlayer}{foreground}
    \node (M) [minimum width=\componentwidth, minimum height=18pt, fill=white, draw=black, thick, anchor=center] at (1.0,1.0) {$D ^{-1}$};
    \end{pgfonlayer}
    \draw [] (1.5,2.0) node [Vertex] {}
        to (1.5,1.0);
    \draw [fill=\fillA] (0.5,-0.75)
        to (0.5,1.5)
        to [out=up, in=up, looseness=2] (-0.5,1.5)
        to (-0.5,-0.75);
    \draw (1.5,1.6 |- M.south)
        to [out=down, in=up, looseness=1.5] (1.5,0.5)
        to [out=down, in=down, looseness=2] (2.5,0.5)
        to (2.5,2.75);
\end{tikzpicture}
\end{aligned}
\end{equation}
\begin{equation}
=\quad
\begin{aligned}
\begin{tikzpicture}[scale=\diagramscale, thick, yscale=-1]
    \draw [fill=\fillA] (0.5,-0.75)
        to (0.5,1.5)
        to [out=up, in=up, looseness=2] (-0.5,1.5)
        to (-0.5,-0.75);
    \draw (1.5,1.0) node [Vertex] {}
        to (1.5,0.5)
        to [out=down, in=down, looseness=2] (2.5,0.5)
        to (2.5,2.75);
\end{tikzpicture}
\end{aligned}
\quad=\quad
\begin{aligned}
\begin{tikzpicture}[scale=\diagramscale, thick, yscale=-1]
    \draw [fill=\fillA] (0.5,-0.75)
        to (0.5,0.25)
        to [out=up, in=up, looseness=2] (-0.5,0.25)
        to (-0.5,-0.75);
    \draw (0,1.5) node [Vertex] {}
        to (0,2.75);
\end{tikzpicture}
\end{aligned}
\end{equation}
This completes the proof.
\end{proof}

\subsection{Secret sharing}

We can represent correctness of a secret sharing procedure in the following way:
\begin{equation}
\label{eq:explicitdensecoding}
\begin{aligned}
\begin{tikzpicture}[scale=\diagramscale, thick]
    \begin{pgfonlayer}{foreground}
    \node (U) [minimum width=\componentwidth, minimum height=18pt, fill=white, draw=black, thick] at (2,2) {$E$};
    \node (M) [minimum width=\componentwidth, minimum height=18pt, fill=white, draw=black, thick, anchor=center] at (1.0,1.0) {$D$};
    \end{pgfonlayer}
    \draw [draw=none, fill=\fillA] (0.5,-0.5) rectangle (-0.5,3);
    \draw [fill=\fillA, draw=none, thick] (1.5,3)
        to (1.5,2)
        to (2.5,2)
        to (2.5,3);
    \draw (0.5,-0.5) to (0.5,3.0);
    \draw (2.5,3)
        to (2.5,0.5)
        to [out=down, in=down, looseness=2] (1.5,0.5)
        to (1.5,3);
\end{tikzpicture}
\end{aligned}
\quad=\quad
\begin{aligned}
\begin{tikzpicture}[scale=\diagramscale, thick]
    \draw [fill=\fillA, draw=none] (-0.5,3.5)
        rectangle (2.5,-0.5);
    \draw (2.5,-0.5) to (2.5,3.5);
    \draw [fill=white] (0.5,3.5)
        to (0.5,1.5)
        to [out=down, in=down, looseness=2] (1.5,1.5)
        to (1.5,3.5);
\end{tikzpicture}
\end{aligned}
\end{equation}
On the left-hand side we begin with some pre-existing public information. This is the information to be communicated by the secret sharing procedure. We prepare two correlated systems forming a one-time pad, and then manipulate the first copy by a procedure $D$ that depends on the value of the classical data. The result is a pair of messages, which are our ciphertexts. Both are then brought together and consumed by a process $E$, producing public information. This process is successful when the result is to copy the original public information.

The important security property of a secret sharing procedure is that if only one ciphertext is available, then no information about the original message can be regained. A strong, constructive way to phrase this is to say that if one of the ciphertexts is erased, the other becomes uniformly random, and independent of the original message. This gives two conditions, with the following graphical representations:
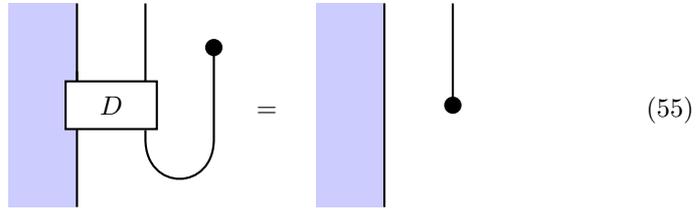
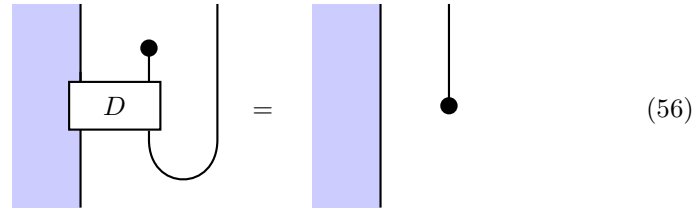
\begin{equation}
\label{eq:secretsharing:security1}
\begin{aligned}
\begin{tikzpicture}[scale=\diagramscale, thick]
    \begin{pgfonlayer}{foreground}
    \node (M) [minimum width=\componentwidth, minimum height=18pt, fill=white, draw=black, thick, anchor=center] at (1.0,1.0) {$D$};
    \end{pgfonlayer}
    \draw [fill=\fillA, draw=none] (0.5,2.5) rectangle (-0.5,-0.5);
    \draw (0.5,2.5)
        to (0.5,1.0)
        to (1.5,1.0)
        to (1.5,2.5);
    \draw (0.5,-0.5) to (0.5,1.5);
    \draw (1.5,1.6 |- M.south)
        to (1.5,0.5) to [out=down, in=down, looseness=2] (2.5,0.5) to (2.5,1.85) node [Vertex] {};
\end{tikzpicture}
\end{aligned}
\quad=\quad
\begin{aligned}
\begin{tikzpicture}[scale=\diagramscale, thick]
    \draw [fill=\fillA, draw=none] (0.5,2.5) rectangle (-0.5,-0.5);
    \draw (0.5,-0.5) to (0.5,2.5);
    \draw (1.5,2.5)
        to (1.5,1.0) node [Vertex] {};
\end{tikzpicture}
\end{aligned}
\end{equation}
\begin{equation}
\label{eq:secretsharing:security2}
\begin{aligned}
\begin{tikzpicture}[scale=\diagramscale, thick]
    \begin{pgfonlayer}{foreground}
    \node (M) [minimum width=\componentwidth, minimum height=18pt, fill=white, draw=black, thick, anchor=center] at (1.0,1.0) {$D$};
    \end{pgfonlayer}
    \draw [fill=\fillA, draw=none] (0.5,2.5) rectangle (-0.5,-0.5);
    \draw (0.5,2.5) 
        to (0.5,1.0)
        to (1.5,1.0)
        to (1.5,1.85) node [Vertex] {};
    \draw (0.5,-0.5) to (0.5,1.5);
    \draw (1.5,1.6 |- M.south)
        to (1.5,0.5) to [out=down, in=down, looseness=2] (2.5,0.5) to (2.5,2.5);
\end{tikzpicture}
\end{aligned}
\quad=\quad
\begin{aligned}
\begin{tikzpicture}[scale=\diagramscale, thick]
    \draw [fill=\fillA, draw=none] (0.5,2.5) rectangle (-0.5,-0.5);
    \draw (0.5,-0.5) to (0.5,2.5);
    \draw (1.5,2.5)
        to (1.5,1.0) node [Vertex] {};
\end{tikzpicture}
\end{aligned}
\end{equation}

\noindent
Equation~\eqref{eq:explicitdensecoding} has an identical structure to the quantum dense coding equation given in~\cite{v12-hsqp}. 
\subsection{Key exchange}
\label{keyexchange}

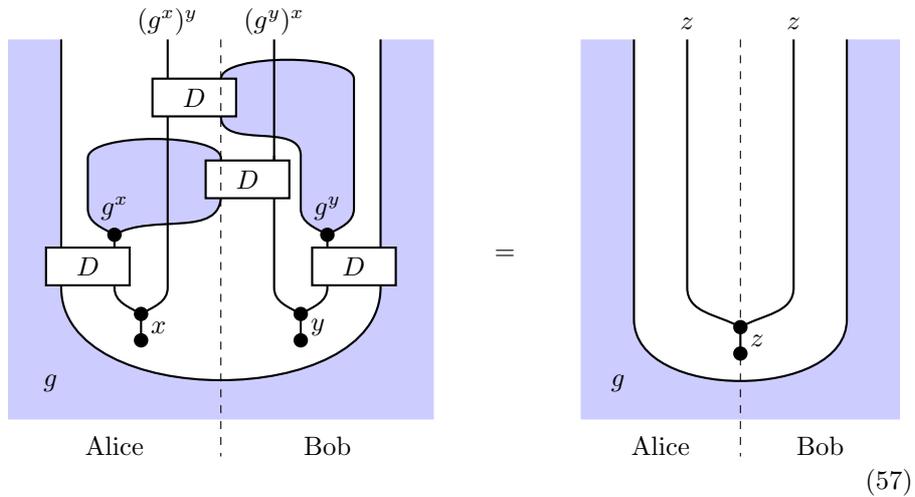
\begin{center}
\tikzset{morphismbox/.style={minimum width=1.1cm, minimum height=0.5cm, fill=white, draw=black, thick}}
\setlength{\vertexradius}{2.3pt}
\begin{equation}
\label{eq:diffie-hellman}
\begin{aligned}
\begin{tikzpicture}[thick, scale=0.7]
\path [use as bounding box] (-4,-3.2) rectangle (4,4.7);
\draw [fill=\fillA, draw=none] (-4,4.7) to (-3,4.7) to (-3,0) to [out=down, in=down] (3,0) to (3,4.7) to (4,4.7) to (4,-2.5) to (-4,-2.5) to (-4,4.7);
\draw (-3,4.7) to (-3,0) to [out=down, in=down] (3,0) to (3,4.7);
\node at (-2,1) [Vertex] {};
\draw (-2,1) to (-2,0) to [out=down, in=\nwangle] (-1.5,-0.5) node [Vertex] {};
\draw (-1.5,-0.5) to node [auto] {$x$} (-1.5,-1) node [Vertex] {};
\draw (2,1) to (2,0) to [out=down, in=\neangle] (1.5,-0.5) node [Vertex] {} to [out=\nwangle, in=down] (1,0) to (1,2.5);
\draw (1.5,-0.5) to node [auto] {$y$} (1.5,-1) node [Vertex] {};
\draw [fill=\fillA] (-2,1) to [out=\neangle, in=down] (0,1.75) to (0,2.45) to [out=up, in=up, looseness=0.5] (-2.5,2.45) to (-2.5,1.5) to [out=down, in=\nwangle] (-2,1);
\node at (2,1) [Vertex] {};
\draw [fill=\fillA] (2,1) to [out=\nwangle, in=down] (1.5,1.5) to (1.5,2.5) to [out=up, in=down] (0,3.25) to (0,3.95) to [out=up, in=up, looseness=0.5] (2.5,3.95) to (2.5,1.5) to [out=down, in=\neangle] (2,1);
\draw (-1.5,-0.5) to [out=\neangle, in=down] (-1,0) to (-1,4.7);
\draw (1,2.4) to (1,4.7);
\node at (-2.5,0.4) [morphismbox] {$D$};
\node at (2.5,0.4) [morphismbox] {$D$};
\node at (0.5,2.05) [morphismbox] {$D$};
\node at (-0.5,3.6) [morphismbox] {$D$};
\node at (-3.2,-1.8) {$g$};
\node at (-2,1.5) {$g^x$};
\node at (2,1.5) {$g^y$};
\node [anchor=south] at (1,4.6) {$(g^y)^x$};
\node [anchor=south] at (-1,4.6) {$(g^x)^y$};
\draw [thin, dashed] (0,4.7) to (0,-3.2);
\node at (-2,-3) {Alice};
\node at (2,-3) {Bob};
\end{tikzpicture}
\end{aligned}
\qquad = \qquad
\begin{aligned}
\begin{tikzpicture}[thick, scale=0.7]
\path [use as bounding box] (-3,-3.2) rectangle (3,4.7);
\draw [fill=\fillA, draw=none] (-3,4.7) to (-2,4.7) to (-2,-0.6) to [out=down, in=down] (2,-0.6) to (2,4.7) to (3,4.7) to (3,-2.5) to (-3,-2.5) to (-3,4.7);
\draw (-2,4.7) to (-2,-0.6) to [out=down, in=down] (2,-0.6) to (2,4.7);
\node at (-2.3,-1.8) {$g$};
\draw [thin, dashed] (0,4.7) to (0,-3.2);
\node at (-1.5,-3) {Alice};
\node at (1.5,-3) {Bob};
\draw (-1,4.7) to (-1,0) to [out=down, in=\nwangle] (0,-0.75) node [Vertex] {} to [out=\neangle, in=down] (1,0) to (1,4.7);
\draw (0,-0.75) to node [auto] {$z$} (0,-1.25) node [Vertex] {};
\node [anchor=south] at (1,4.7) {$z$};
\node [anchor=south] at (-1,4.7) {$z$};
\end{tikzpicture}
\end{aligned}
\end{equation}
\vspace{-9pt}
\end{center}

Our final study is Diffie-Hellman key exchange~\cite{dh76-ndc}, a procedure whereby two parties who share common public base information can obtain a shared secret key by exchanging only public information. The bicategorical diagram representing it is given as equation~(\ref{eq:diffie-hellman}). The symmetric monoidal bicategory structure is essential here, as it gives meaning to the overlapping of parts of the diagram.

Ambient public information represents the base $g$ to be used by the protocol. Alice and Bob nondeterministically choose private keys $x$ and $y$ respectively, which they duplicate. They then each apply a controlled operation $D$, which in the conventional implementation depends on public information $p$, and transforms private information as $q \mapsto p^q$ with respect to some fixed cyclic group structure. The result of this is then published and transferred to the other party, where $D$ is applied once again. As a result, both parties share the key $g ^{xy}$.

The protocol is implemented correctly if, neglecting the public data produced during the procedure, the private keys are identical and uncorrelated with the initial base. Erasing the public data is necessary for information-theoretical security in the classical case, and for maintaining coherence in any quantum interpretation.

Our graphical formalism captures this structure in a clear way, which moreover can be used to formally verify correctness of an implementation.


\begin{thebibliography}{10}
\providecommand{\url}[1]{#1}
\csname url@samestyle\endcsname
\providecommand{\newblock}{\relax}
\providecommand{\bibinfo}[2]{#2}
\providecommand{\BIBentrySTDinterwordspacing}{\spaceskip=0pt\relax}
\providecommand{\BIBentryALTinterwordstretchfactor}{4}
\providecommand{\BIBentryALTinterwordspacing}{\spaceskip=\fontdimen2\font plus
\BIBentryALTinterwordstretchfactor\fontdimen3\font minus
  \fontdimen4\font\relax}
\providecommand{\BIBforeignlanguage}[2]{{%
\expandafter\ifx\csname l@#1\endcsname\relax
\typeout{** WARNING: IEEEtran.bst: No hyphenation pattern has been}%
\typeout{** loaded for the language `#1'. Using the pattern for}%
\typeout{** the default language instead.}%
\else
\language=\csname l@#1\endcsname
\fi
#2}}
\providecommand{\BIBdecl}{\relax}
\BIBdecl

\bibitem{g68-ua}
R.~Whitehead, \emph{A Treatise on Universal Algebra}.\hskip 1em plus 0.5em
  minus 0.4em\relax Cambridge University Press, 1898.

\bibitem{l80-flc}
J.~Lambek, \emph{From lambda calculus to Cartesian closed categories}.\hskip
  1em plus 0.5em minus 0.4em\relax Academic Press, 1980, pp. 376--402.

\bibitem{w05-cae}
G.~Washburn, ``Cause and effect: Type systems for effects and dependencies,''
  University of Pennsylvania Department of Computer and Information Science,
  Tech. Rep. MS-CIS-05-05, November 2005.

\bibitem{ac04-csqp}
S.~Abramsky and B.~Coecke, ``A categorical semantics of quantum protocols,''
  \emph{Proceedings of the 19th Annual IEEE Symposium on Logic in Computer
  Science}, pp. 415--425, 2004, iEEE Computer Science Press.

\bibitem{ac08-cqm}
------, \emph{Handbook of Quantum Logic and Quantum Structures}.\hskip 1em plus
  0.5em minus 0.4em\relax Elsevier, 2008, vol.~2, ch. Categorical Quantum
  Mechanics.

\bibitem{hv13-cqm}
C.~Heunen and J.~Vicary, ``Lectures on categorical quantum mechanics,''
  available at http://www.cs.ox.ac.uk/courses/cqm.

\bibitem{bs-rosetta}
J.~Baez and M.~Stay, \emph{Physics, Topology, Logic, and Computation: A Rosetta
  Stone}.\hskip 1em plus 0.5em minus 0.4em\relax Springer, 2011, vol. 813, pp.
  95--172.

\bibitem{bd95-hda}
J.~C. Baez and J.~Dolan, ``Higher-dimensional algebra and topological quantum
  field theory,'' \emph{Journal of Mathematical Physics}, vol.~36, pp.
  6073--6105, 1995.

\bibitem{v12-hsqp}
J.~Vicary, ``Higher semantics of quantum protocols,'' \emph{Proceedings of the
  27th Annual IEEE Symposium on Logic in Computer Science}, 2012.

\bibitem{b93-teleport}
C.~H. Bennett, G.~Brassard, C.~Cr\'epeau, R.~Jozsa, A.~Peres, and W.~K.
  Wootters, ``Teleporting an unknown quantum state via dual classical and
  {{E}instein-{P}odolsky-{R}}osen channels,'' \emph{Physical Review Letters},
  vol.~70, no.~13, pp. 1895--1899, 1993.

\bibitem{p09-qcs}
D.~Pavlovic, ``Quantum and classical structures in nondeterministic
  computation,'' in \emph{Proceedings of Quantum Interaction 2009}.

\bibitem{s13-ccb}
M.~Stay, ``Compact closed bicategories,'' 2013, arXiv:1301.1053.

\bibitem{bs10-rosetta}
J.~C. Baez and M.~Stay, \emph{New Structures for Physics}.\hskip 1em plus 0.5em
  minus 0.4em\relax Springer, 2011, ch. Physics, Topology, Logic and
  Computation: A Rosetta Stone, pp. 95--172.

\bibitem{s11-sgl}
P.~Selinger, \emph{New Structures for Physics}.\hskip 1em plus 0.5em minus
  0.4em\relax Springer, 2011, ch. A Survey of Graphical Languages for Monoidal
  Categories, pp. 289--355.

\bibitem{l06-faaa}
\BIBentryALTinterwordspacing
A.~D. Lauda, ``Frobenius algebras and ambidextrous adjunctions,'' \emph{Theory
  and Applications of Categories}, vol.~16, no.~4, pp. 84--122, 2006. [Online].
  Available: \url{http://www.tac.mta.ca/tac/volumes/16/4/16-04abs.html}
\BIBentrySTDinterwordspacing

\bibitem{hcc12-rfa}
C.~Heunen, I.~Contreras, and A.~Cattaneo, ``Relative frobenius algebras are
  groupoids,'' \emph{Journal of Pure and Applied Algebra}, vol. 217, pp.
  114--124, 2012.

\bibitem{dh76-ndc}
W.~Diffie and M.~Hellman, ``New directions in cryptography,'' \emph{IEEE
  Transactions on Information Theory}, vol.~22, no.~6, pp. 644--654, 1976.

\end{thebibliography}


\end{document}